\documentclass{article}
\usepackage{epsf,graphicx}

\usepackage{amssymb,amsmath}
\usepackage{a4wide}
\usepackage{ulem}
 \usepackage{amsthm}
\usepackage{enumitem}
\usepackage{multirow}
\graphicspath{{./figs/}}%

\newtheorem{definition}{Definition}
\newtheorem{proposition}{Proposition}

\usepackage[font=footnotesize]{caption}

\usepackage{color}
\usepackage[usenames,dvipsnames]{xcolor}

\usepackage{pgfplots}

\usepackage{url}

\def\sr{\approx}
\def\wr{\sim}
\def\notwr{\nsim}

\def\Else{\mbox{\bf else\ }}

\def\For{\mbox{\bf for\ }}
\def\Foreach{\mbox{\bf foreach\ }}

\def\If{\mbox{\bf if\ }}

\def\algfigure#1{
{
\begin{tabbing}
xxx\=xxx\=xxx\=xxx\=xxx\=xxx\=xxx\=xxx\=xxx\=xxx\= \kill
#1
\end{tabbing}
}
}

\title{Group evolution patterns in running races}
\author{Y. Diez\thanks{Faculty of Science, Yamagata University, Yamagata, Japan} \and M. Fort\thanks{IMAE Department, University of Girona, Girona, Spain} \and M. Korman\thanks{Department of Computer Science, Tufts University, Medford, USA.} \and J.A. Sellar\`{e}s\thanks{LSI, Universitat Polit\`{e}cnica de Catalunya, Barcelona, Spain}}
\date{}

\begin{document}
\maketitle

\begin{abstract}

We address the problem of tracking and detecting interactions between the different groups of runners that form during a race. In athletic races control points are set to monitor the progress of athletes over the course. Intuitively, a {\it group} is a sufficiently large set of athletes that cross a control point together. After adapting an existing definition of group to our setting we go on to study two types of group evolution patterns. The primary focus of this work are {\it evolution patterns}, i.e. the transformation and interaction of groups of athletes between two consecutive control points. We provide an accurate geometric model of the following evolution patterns: survives, appears, disappears, expands, shrinks, merges, splits, coheres and disbands, and present algorithms to efficiently compute these patterns. Next, based on the algorithms introduced for identifying evolution patterns, algorithms to detect {\it long-term patterns} are introduced. These patterns track global properties over several control points: surviving, traceable forward, traceable backward and related forward and backward. Experimental evaluation of the algorithms provided is presented using real and synthetic data. Using the data currently available, our experiments show how our algorithms can provide valuable insight into how running races develop. Moreover, we also show how, even if dense (synthetic) data is considered, our algorithms are also able to process it in real time.

\end{abstract}

{\bf  keywords} Computer science; Information system; Running race analysis; Groups in running races; Evolution patterns; Long-term patterns.

\section{Introduction}

Long-distance races involve a large number of runners, many of whom carry some kind of timing system. Aside from this, every athlete carries a small lightweight chip that records their race bib number and the exact times they cross certain designated control points such as the start line, the finish line and other locations in-between. The athletes themselves also often gather data on their own. Nowadays many of them wear GPS-enabled watches or carry smartphones that record their precise position every few seconds. In this paper we consider using this information to track and detect the interaction between the different groups of runners that form during long-distance races such as 5K, 10K, half and full marathons. To the best of our knowledge, this is the first work to address this problem.

Both types of data are already widely used; control point information is used by race organizers to track in-race athletes, detect possible problems or cheating as well as provide runners with information about their performance in the race. Personal data is used by both professional and amateur runners  during training and provides a very detailed picture of how a run developed. Specifically, these devices provide an example of the type of data that is of interest for runners because apart from recording the GPS position at every instant, data is processed to produce other magnitudes such as average and instantaneous speed, pace (minutes needed to run one kilometer or mile) and even running cadence (number of steps per minute).

One type of data is when the location fixed and the times when runners cross control points are tracked, while the other is when, a wearable device measures the location of the runner at regular intervals. We believe that there will be a natural tendency for the differences between the two types to disappear as the frequency in sampling grows. For the sake of simplicity from here on we will assume that data is measured at control points, but we note that our definitions and algorithms naturally extend to the other case. Only minor modifications which would have no impact at all in the key properties or runtime of the algorithms would be needed.

In any case, as location data for runners increases in popularity what is becoming clear is that such data contain valuable and implicit knowledge. Thus, the goal is to make this knowledge efficiently explicit in an intelligible form. The algorithms and definitions presented in this paper can provide meaningful insight for runners, race organizers and spectators alike. Because the algorithms are fast, they can be used in real-time during a race and/or as a tool for analysis after the race is over.

\subsection{Our contribution}
We study {\it evolution patterns} i.e. the evolution and interaction of groups of athletes that typically happen during races which, in our case, occur at two consecutive control points. From the records of their bib numbers and their running times at specific control points we can identify specific groups of athletes.

For example, when the race starts it is usual for most runners to be together. This forms a large group that will split into several smaller groups of people with similar running speeds. In addition to this split operation we consider other evolution patterns such as: survives, appears, disappears, expands, shrinks, merges, coheres and disbands. For each of these patterns we provide accurate geometric modeling that we use directly in the design of efficient algorithms for detecting them. Our models include the joint evolution of more than one group, thus allowing for insights on group behavior for the whole set of athletes.

Analyzing the behavior of the groups of athletes throughout the race is also interesting. Thus, we study {\it long-term patterns}, i.e. the evolution and interaction of the groups of athletes appearing at more than two consecutive control points. We define four long-term patterns: surviving, traceable forward, traceable backward and related forward or backward. We provide algorithms based on the algorithms introduced for the evolution patterns, to detect the different long-term patterns, to report their length - the number of consecutive control points where they appear - and to determine the largest length sequence of consecutive control points where each long-term pattern appears over the course of the race.

Specifically, in this paper we:

\begin{itemize}

    \item Introduce a formalization for the trajectories of runners in long-distance running races. This formalization builds on existing work in the computational geometry community. The fact that these races have fixed courses allow us to consider the trajectories to be one dimensional. 
    
    \item Present formal definitions of group and evolution patterns. We introduce algorithms to find groups and evolution patterns, all of which are asymptotically optimal.

    \item We study long-term patterns to track information that spans more than two control points. We provide asymptotically optimal algorithms to detect and report these patterns.

    \item We provide an extensive evaluation, using real and synthetic data, that shows how our algorithms perform in practice. This evaluation shows how our algorithms can, in very few seconds, process existing (sparse) data for a marathon with over thirty thousand runners. Since denser data is not yet available, we create a much denser (synthetic) data of roughly 2.5 million runners, and observe that our algorithms can still work in real-time.

    \end{itemize}

The definitions and algorithms presented in this paper can provide meaningful insights for runners, race organizers and spectators alike. The algorithms presented are fast and can therefore be used in real-time during a race as well as an analysis tool once the race is finished.

\subsection{Paper organization}

Section \ref{relatedwork} provides an overview of previous related work. Section \ref{preliminaries} presents the concepts that we use to define group evolution patterns amongst which the {\it inclusion coefficient} is the most important. Section \ref{grouping} contains the formal definition of a group and how we measure the relationship between groups in our setting. The group evolution patterns (survives, appears, disappears, expands, shrinks, merges, splits, coheres and disbands) are defined in Section~\ref{patterns}. Section~\ref{evolutionGraphs} introduces the {\it evolution graphs}. An evolution graph is built using a relation between groups at a pair of consecutive control points defined by a thresholded version of the inclusion coefficient. The evolution graphs are used as a tool to detect any group evolution patterns. Section \ref{algorithms} presents the algorithms for computing group evolution patterns and their asymptotic costs. Next, Section \ref{problems} provides the algorithms used to detect long-term patterns over several control points (surviving, traceable forward, traceable backward and related). These algorithms also compute the length of the patterns detected and the largest length sequence of consecutive control points where each long-term pattern appears. In Section \ref{experiments}, the practical performance of all the algorithms presented is described experimentally with both real and synthetic data. Finally, conclusions are presented in Section \ref{conclusions}.

\section{Related work} \label{relatedwork}

Trajectory databases, in many cases rather large in volume and complex in structure, contain valuable and implicit knowledge that can be extracted using geometric analysis and data mining techniques \cite{BDDM17,LCI09, SHJJ16, Zhe15}.
Extensive research has been done on the problem of detecting sets of entities moving together over a period of time. The main concepts used are: flocks \cite{LI02,GK06,BGHW08,VBT09,FSV14}, moving clusters \cite{KMB05,SNTS06,CMC14}, herds \cite{HCD08}, convoys \cite{JYZJS08}, swarms \cite{LDHK10}, and groups \cite{BBKSS15,KLSW16}.

Although these concepts differ slightly from each other (see definitions below), a major trait they all have in common is whether or not two entities will belong to the same set at some instant, that will depend not only on their current distance, but also on their past and future distances. This is motivated by the fact that the space is two dimensional (or higher), and so even though two entities are nearby, they are possibly traveling along different routes. Hence, for robustness, the requirement that the entities are close for a long period of time is added. This problem, however, does not happen in races since everyone is going in the same direction along a predefined track (i.e., our space is one-dimensional). We take advantage of this fact, and look at each control point independently to determine whether or not two entities are related around that time period.

The fact that the runners move on a predefined course allows us to define efficient algorithms to track all possible events. For example, we can detect that a group has mostly remained the same between two control points. In most of the previous work, even when a single entity leaves a group the whole group is considered to have disappeared (and possibly a new one without the entity has been created). Instead, we allow a small portion of the runners to change from control point to control point, thus giving greater robustness to our algorithms.

Indeed, without our one-dimensional property, this and other problems become extremely difficult, if not impossible, to solve in reasonable time (for example, finding the largest group in which we allow a single runner to temporarily leave the group is NP-hard~\cite{BBKSS15}). 

\paragraph{Other group definitions}
        A flock is a set of entities that move together within a disc of some user-specified size \cite{LI02}.
        Benkert {\it et al.} \cite{BGHW08} proposed a more realistic definition of a flock, where a minimum number of consecutive time-steps are considered, and presented an efficient approximation algorithm for detecting and reporting flocks.
        Using this model, Gudmundsson and van Kreveld \cite{GK06} presented approximation algorithms to find the flock that is together for the longest period.
        Vieira {\it et al.} ~\cite{VBT09} gave a characterization of the potential flocks for every time-step and proposed several heuristic strategies to discover maximal flock patterns with a predefined time duration.
        Fort {\it et al.} \cite{FSV14} studied the problem of finding flock patterns and presented a parallel GPU-based algorithm for reporting all maximal flocks, the largest flock and the flock of maximum duration.

        Kalnis {\it et al.} \cite{KMB05} presented an algorithm for detecting moving clusters, where clustered entities at consecutive time-steps share a large number of common entities, but entities may join and leave during the lifetime of the cluster. In \cite{SNTS06}, a framework for modeling and detecting changes of clusters at different time-steps
        (appears, disappears, survives, splits, absorbs) is proposed.
        T.~L. Coelho da Silva  {\it et al.} \cite{CMC14} developed a method that performs density-based clustering on
        trajectory data at regular time-steps, and analyzed cluster evolution (appear, disappear,
        expand, shrink, split, merge and survive).

        The concept of herd, that relies on the notion of F-score to identify cluster overlaps at consecutive time-steps, was introduced in \cite{HCD08}. Moreover, four types of herd evolvements (expand, join, shrink, and leave) were studied.

        Jeung {\it et al.} in \cite{JYZJS08} proposed the notion of convoy, that uses density connectedness for spatial clustering.
        Aung and Tan \cite{AT10} introduced the notion of evolving convoys to better understand the states of convoys.
        The concept of swarm \cite{LDHK10}, which also uses density connectedness for spatial clustering, permits moving entities to travel together for a number of nonconsecutive time-steps.

        Buchin {\it et al.} in \cite{BBKSS15} introduced a formal definition of group, that relies on three parameters (distance between entities, group duration and group size), and uses the notion of $\varepsilon$-connectedness. They analyzed the mathematical structure of a group, and presented efficient algorithms for computing all maximal groups in a given set of trajectories. M.~van Kreveld {\it et al.}. in \cite{KLSW16} redefined the previous definition and argued that the new definition corresponds better to human intuition, particularly in dense environments. In particular, they provided an algorithm for trajectories moving in $\mathbb{R}^1$ that computes all maximal groups.

        This is an application-driven research area where most contributions are motivated by real problems and algorithms are frequently implemented and experimentally analyzed. These analyses are performed either with synthetic or real data. Both the number of trajectories considered and the number of time-steps that each trajectory has, vary depending on the specific application that motivates the research. For example, \cite{FSV14} reported flock patterns by using hundreds of dense trajectories (tens of thousands of time-steps) corresponding to buses, trucks and human subjects, while \cite{CMC14} used thousands of non-dense (tens of time-steps) to analyze related vehicle data. In terms of group computations, the only previous implementation (to the best of our knowledge) was presented in \cite{BBKSS15} and worked with two datasets. The first contained 400 synthetic trajectories with 818 time-steps each, the second dataset was made up of data from migrating animals and contained 126 trajectories with 1264 time-steps each.

        In terms of the use of GPS data from runners, many commercial tools exist to visualize the course of a run on a map and obtain magnitudes such as average pace or speed. Most of these tools are tied to particular GPS watches or mobile phone apps and seem to be primarily aimed at providing an aid for training journals and, in some cases, providing an element of social network interaction for runners living near each other. A particularly interesting example is The Clusterer project from Strava Labs \cite{Strava}, which,  groups together activities in terms of geographical location, distance completed, activity type (running, walking, biking, etc.) and so on.

\section{Preliminaries}\label{preliminaries}

    In this section, we introduce the {\it inclusion coefficient}, which is a tool used to measure the relation between two (abstract) sets. We then show how we use it to determine relationships between groups of runners.

    Given two sets $(A$ and $B)$ the {\it inclusion coefficient} $I(A,B)$ is the ratio:

    $$I(A,B) = {{|A \cap B|}\over{|A|}}\,.$$

    The inclusion coefficient measures the proportion of elements of $A$ contained in $B$ and has, among others, the following properties:

    \begin{enumerate}[label=\alph*)]
        \item in general $I(A,B) \ne I(B,A)$, i.e. the inclusion coefficient is not symmetric.
        \item $0 \le I(A,B)\le 1$ because $A \cap B \subseteq A$ and is $0 \le |A \cap B| \le |A|$.
        \item $I(A,B)=0 \Leftrightarrow |A \cap B| = 0 \Leftrightarrow A \cap B = \emptyset $.
        \item  $I(A,B)=1 \Leftrightarrow |A \cap B|=|A| \Leftrightarrow A \cap B=A \Leftrightarrow  A \subseteq B $.
        \item  $I(A,B)=1/2 \Leftrightarrow |A \cap B|/|A|=1/2 \Leftrightarrow 2|A \cap B|=|A| \Leftrightarrow |A \cap B|=|A| - |A \cap B|$.
        \item   $I(A,B)>1/2 \Leftrightarrow |A \cap B|/|A|>1/2 \Leftrightarrow 2|A \cap B|>|A| \Leftrightarrow |A \cap B|>|A| - |A \cap B|$.
        \item  $I(A,B)<1/2 \Leftrightarrow |A \cap B|/|A|<1/2 \Leftrightarrow 2|A \cap B|<|A| \Leftrightarrow |A \cap B|<|A| - |A \cap B|$.
    \end{enumerate}

    Next, bearing properties e), f) and g) in mind, we provide some definitions that relate  $|A| - |A \cap B|$, the number of elements of $A$ not contained in $B$, and $|A \cap B|$, the number of elements of $A$ contained in $B$. The goal is to adapt the inclusion coefficient to the cases of interest in our application as well as provide general interpretations of its values that assume an intuitive meaning in practice.

    From now on, we pick a parameter $\mu\in(1/2,1]$ that represents how strict we are in the resemblance between two groups of runners.

    \begin{definition}
        Sets $A$ and $B$ are \textbf{weakly related}, denoted $A \wr B$, if and only if $I(A,B) \ge \mu.$
    \end{definition}

    When $A \wr B$, we say that a sufficiently large part of $A$ is contained in $B$ (or that many elements of $A$ are in $B$). We will also use the notation $A \notwr B$ to denote that $A$ and $B$ are not weakly related.

    \begin{definition}
        Sets $A$ and $B$ are \textbf{strongly related}, denoted $A \sr B$, if and only if $I(A,B) \ge \mu$ and $I(B,A) \ge \mu$.
    \end{definition}

    Consequently, $A \sr B$ if and only if $A \wr B$ and $B \wr A$. Thus, when $A \sr B$, we say that a sufficiently large part of $A$ is contained in $B$ and a sufficiently large part of $B$ is contained in $A$ (or that many elements of $A$ are in $B$ and many elements of $B$ are in $A$).

\subsection{Thresholded inclusion coefficient and inclusion coefficient of unions of disjoint sets}

    Next, we provide some properties of the inclusion coefficient that we will use throughout the remainder of this paper.

    \begin{proposition}
        Let $A$, $B$ and $B'$ be three sets such that $A \wr B$ and $A \wr B'$. Then, it holds that $B\cap B' \neq \emptyset$.
        \label{Proposition:onlyoneset}
    \end{proposition}
    \begin{proof}
        Since $$I(A,B) = \frac{|A \cap B|}{|A|} \ge \mu > \frac{1}{2},$$ then $$ |A \cap B| > \frac{|A|}{2} .$$
        For any set $B'$ disjoint with $B$ we have $$A \cap B' \subseteq A \setminus (A \cap B).$$
        Hence $$|A \cap B'| \le |A \setminus (A \cap B)| = |A| - |A \cap B| < |A| - \frac{|A|}{2} = \frac{|A|}{2}$$ and $$I(A,B') = \frac{|A \cap B'|}{|A|} < \frac{1}{2} < \mu.$$
    \end{proof}

    \begin{proposition}
        Let $A_1, \cdots ,A_k$ be pairwise disjoint sets, $A_i\cap A_j=\emptyset, \forall i,j \in \{1\cdots k\}$, then
        $$I(\cup_{i=1}^{k} A_i,B)=\sum\limits_{i=1}^{k}\frac{|A_i| \, I(A_i,B)}{ \sum_{i=1}^{k} |A_i|}\, .$$
        \label{Proposition:union1}
    \end{proposition}
    \begin{proof}
        \begin{eqnarray*}
            I(\cup_{i=1}^{k} A_i,B)&=& \frac{|(\cup_{i=1}^{k} A_i) \cap B|}{|\cup_{i=1}^{k} A_i|} = \frac{|\cup_{i=1}^{k} (A_i \cap B) |}{|\cup_{i=1}^{k} A_i|}
        \end{eqnarray*}
        Since the sets $A_i$ are pairwise disjoint,
        \begin{eqnarray*}
            = \frac{\sum_{i=1}^{k} |A_i \cap B|}{\sum_{i=1}^{k} |A_i|} =|B|  \sum_{i=1}^{k} \frac{ I(B,A_i)}{\sum_{i=1}^{k} |A_i|}
        \end{eqnarray*}
        Where the last inequality follows from the fact that $I(B,A_i) = {{|B \cap A_i|}\over{|B|}}$. Finally, we use the fact that $I(B,A_i) = \frac{|A_i|}{|B|} I(A_i,B) $ to obtain
        \begin{eqnarray*}
         = \sum_{i=1}^{k} \frac{ |A_i| I(A_i,B)}{\sum_{i=1}^{k} |A_i|}
        \end{eqnarray*}
        as claimed.
    \end{proof}

    \begin{proposition}
        Let $A_1, \cdots ,A_k$ be $k$ pairwise disjoint sets (that is, $A_i\cap A_j=\emptyset, \forall i,j \in \{1\cdots k\}$). If $A_i \wr B$ $\forall i \in \{1\cdots k\}$, then it holds that $\cup_{i=1}^{k} A_i \wr B$.
        \label{Proposition:union2}
    \end{proposition}
    \begin{proof}
        From Proposition \ref{Proposition:union1}, we have
        $$I(\cup_{i=1}^{k} A_i, B) =\sum\limits_{i=1}^{k}\frac{|A_i| \, I(A_i,B)}{ \sum_{i=1}^{k} |A_i|}
        \ge \sum\limits_{i=1}^{k}\frac{|A_i| \, \mu}{ \sum_{a=1}^{k} |A_i|}= \frac{\mu}{ \sum_{a=1}^{k} |A_i|} \sum\limits_{i=1}^{k}|A_i| = \mu.$$
    \end{proof}

\section{Grouping athletes} \label{grouping}
    Let $\mathcal{E}=\{ e_1, \cdots , e_n \}$ denote the set of $n$ runners and $\mathcal{X}=\{ x_1, \cdots , x_{\sigma} \}$ denote the set of $\sigma$ control points sorted spatially over the course. The race summary of athlete $e_i \in \mathcal{E}$ is represented by the list of pairs in the form of $(x_1,t_{1,i}), \cdots,(x_{\sigma},t_{\sigma,i})$, and  sorted in accordance with the order of the control points, where $t_{j,i}$ denotes the time the athlete $e_i$ passes through the control point $x_j$.  Thus, the times of the athlete recorded at subsequent control points are also sorted: $t_{j,i}<t_{k,i}$ if and only if $j<k$.

    Notice that for each race summary, each of the constituent points contains two real values corresponding to time and space. Although most of the literature introduces the "group" concept as a function of space (that is, a number of athletes that are nearby), in this paper we focus on time instead. This is done basically to conform naturally to the way running data is expressed in races (i.e. we receive the information when the runners reach the designated control points). However, the role of the two parameters is symmetric, and thus it can be modified to the other case easily.
    This represents a variation of the definition of a group of moving entities provided in~\cite{KLSW16}. Consequently, in the following we adapt the definition in the aforementioned reference to match our formalization.

    Let $\varepsilon$ be a predefined threshold. The $\varepsilon$-interval of an athlete $e_i$ at position $x_j$ is the temporal interval $I_{j,i,\varepsilon}=[t_{j,i}-\varepsilon/2, t_{j,i}+\varepsilon/2]$. Two athletes, $e_1$ and $e_2$ are considered to be {\it directly connected} at position $x_j$ if and only if $|t_{j,1}-t_{j,2}|<\varepsilon$ or, equivalently, if and only if $I_{j,1,\varepsilon} \cap I_{j,2,\varepsilon} \ne \emptyset$ (intuitively speaking, the two runners crossed the same control point less than $\varepsilon$ units of time apart from each other).

    Given a subset $S \subseteq \mathcal{E}$ of athletes, {\it two athletes} $e$ and $e'$ are {\it $\varepsilon$-connected in $S$} at position $x_j$ if there is a sequence $e = e_0, \cdots, e_k = e'$ of athletes in $S$ such that, for all $i$, $e_i$ and $e_{i+1}$ are directly connected at $x_j$. A subset $S \subseteq \mathcal{E}$ of athletes is {\it $\varepsilon$-connected in $S$} at position $x_j$ if all athletes in $S$ are pairwise $\varepsilon$-connected in $S$ at position $x_j$. This means that the union of the $\varepsilon$-intervals of athletes in $S$ forms the interval $[t_{j,f}-\varepsilon/2,t_{j,l}+\varepsilon/2]$, where $e_f$ and $e_l$ denote the first and last athletes of $S$ that cross position $x_j$, respectively.

    The set $S$ forms a {\it component} at position $x$ if and only if $S$ is $\varepsilon$-connected in  $S$, and $S$ is maximal with respect to this property. The set of components $C(x)$ at position $x$ forms a partition of the athletes in $\mathcal{E}$ at position $x$.

    A {\it group} at position $x$ is a maximal set $S$ of $\varepsilon$-connected athletes in $S$ that contains a minimum number $m$ of athletes (that is, $|S| \ge m$ and $m$ is some prefixed constant). We denote $G(x)$ the set $\{S_1, \cdots, S_t\}$ of groups at position $x$ and $O(x) = \mathcal{E} - \bigcup_{S_i \in G(x)}$ the set of outliers at position $x$. Observe that at each control point each athlete belongs to at most one group. Consequently, the number of groups $t \leq n/m$.

    \begin{proposition}
        For any group $S_i \in G(x)$ there cannot exist two groups $S'_j,S'_k \in G(x')$ ($S'_j\neq S'_k$) such that $S_i \wr S'_j$, $S_i \wr S'_k$.
        \label{Pro:onlyonegroup}
    \end{proposition}
    \begin{proof}
        Since the sets $S'_j$ and $S'_k$ are disjoint, Proposition \ref{Proposition:onlyoneset} and $S_i \wr S'_j$ imply that $I(S_i,S'_k) < \frac{1}{2} < \mu$, and thus $S_i$ and $S'_k$ cannot be weakly related.
    \end{proof}

    Observe that this proposition is asymmetric and thus, only relates to the weak relation {\it in one direction}. That is, once we have  $S_i \wr S'_j$ we cannot have the same set being related to any other $S'_k$. However, this does not imply anything in terms of weak relations in the opposite direction. In particular, many groups $S_{\ell} \in G(x)$ with $S_{\ell} \wr S'_j$ may exist. Notice also that, this type of behavior does not happen with the strong (bidirectional) relation.

    \begin{proposition}
        For any group $S_i \in G(x)$ there cannot exist two groups $S'_j,S'_k \in G(x')$ such that $S_i \sr S'_j$, $S_i \sr S'_k$, and $S'_j\neq S'_k$. Moreover, for any group $S'_i \in G'(x)$ there cannot exist two groups $S_j,S_k \in G(x)$ such that $S_j \sr S'_i$, $S_k \sr S'_i$, and $S_j\neq S_k$.
        \label{Pro:onlyonegroupBidirecccional}
    \end{proposition}

    Proof of this claim is identical to the proof of Proposition~\ref{Proposition:onlyoneset}. Observe that, although you can only have one strong relation $S_i \sr S'_j$, other groups could be weakly related with either of the two groups. That is,$S_{\ell}\in G(x)$ with $S_{\ell} \wr S'_j$ and $S'_k\in G(x')$ with $S'_k \wr S_i$ may exist.

\section{Group evolution patterns}\label{patterns}

    We now formally define the group evolution patterns. We consider the following patterns: survives, appears, disappears, expands, shrinks, merges, splits, coheres and disbands.

    \begin{figure}
        \begin{center}
        \includegraphics[width=0.8\textwidth]{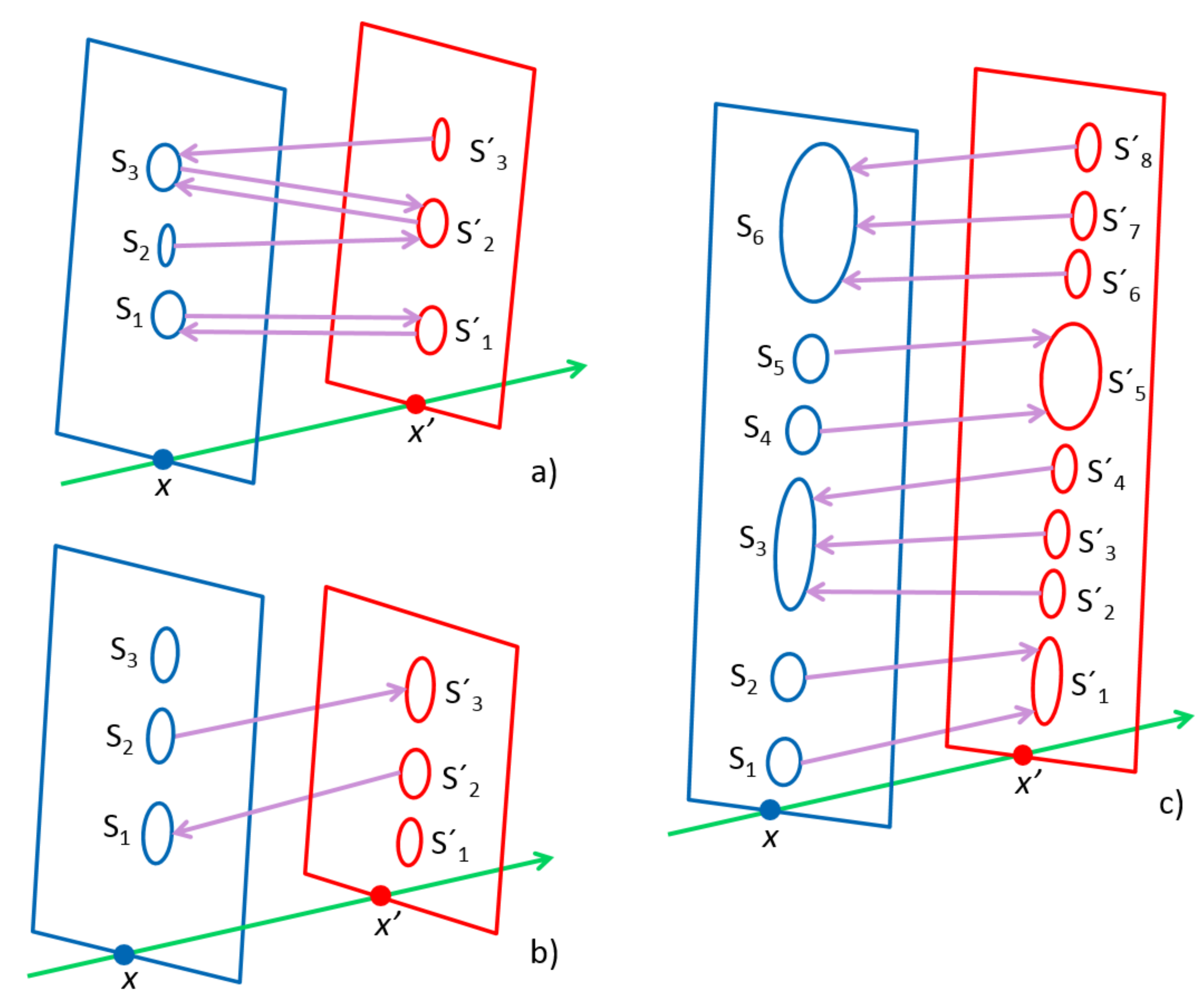}
        \end{center}
        \caption{Examples of all possible group evolution patterns a) group $S_1$ survives as group $S'_1$; group $S_3$ survives as group $S'_2$ and, moreover, group $S'_2$ absorbs group $S_2$ and group $S_3$ spawns group $S'_3$;
        b) group $S'_1$ appears and group $S_3$ disappears; group $S_2$ expands to group $S'_3$ and group $S_1$ shrinks into group $S'_2$.
        c) groups $S_1$ and $S_2$ merge into group $S'_1$ and group $S_3$ splits into groups $S'_2$, $S'_3$ and $S'_4$;
        groups $S_4$ and $S_5$ cohere into group $S'_5$ and group $S_6$ disbands into groups $S'_6$, $S'_7$ and $S'_8$.}
        \label{fig:EvolGraph1}
    \end{figure}

    \begin{description}

        \item[{\bf Survives}] Intuitively speaking, a group at a control point {\it Survives} as a group at the next control point when the two groups differ only by a few athletes. 
            Additionally, a group that {\it survives} may absorb one or more groups and/or spawn one or more groups. Notice that these two latter behaviors might happen simultaneously.

            Formally, group $S_i \in G(x)$ {\it survives}  as group $S'_j \in G(x')$, if and only if $S_i \sr S'_j$
            (see Figure \ref{fig:EvolGraph1} a).

            We say that a group $S_i$ that survives as group $S'_j \in G(x')$ also {\it absorbs}  groups  $S_{i_1}, \cdots, S_{i_t} \in G(x)$, different from $S_i$, if and only if $S_{i_k} \wr S'_j$, $1 \le k \le t$. Thus, even though $S_i$ and $S'_j$ are strongly related (intuitively speaking, they are the same group except for a few runners that changed) the group $S'_j$ incorporates most of the runners of the smaller groups $S_{i_k}$.

            Conversely, we say that a group $S_i$ that survives as group $S'_j \in G(x')$ {\it spawns} one or more groups $S'_{j_1}, \cdots, S'_{j_r} \in G(x')$, different from $S'_j$, if and only if $S_{j_{\ell}} \wr S_i$, $1 \le \ell \le r$. See an example of a group that survives (with spawn and absorb) in Figure \ref{fig:EvolGraph1} a).

        \item[{\bf Appears}] A group {\it Appears} at control point $x'$ when it has no previous relationship with groups of the previous control point $x$. Formally, group $S'_j \in G(x')$ {\it appears}, if and only if $\forall S_i \in G(x)$ it holds that $S_i \notwr S'_j$ and $S'_j \notwr S_i$ (see Figure \ref{fig:EvolGraph1} b)).

        \item[{\bf Disappears}] The opposite of appears, we say that a group {\it Disappears} at control point $x$ when it has no relationship with the groups at the next control point $x'$ (see Figure \ref{fig:EvolGraph1} b)). Formally, group $S_i \in G(x)$ {\it disappears} , if and only if $\forall S'_j \in G(x')$ it holds that $S_i \notwr S'_j$ and $S'_j \notwr S_i$.

        \item[{\bf Expands}] A group at control point $x$ {\it Expands} to the next control point when it has grown so much that many athletes in the new group did not belong to the original group (despite the new group containing most of the athletes from the original group, see Figure \ref{fig:EvolGraph1} b)). Formally, group $S_i \in G(x)$ {\it expands} into group $S'_j \in G(x')$ if and only if:
            \begin{enumerate}
                \item $S_i \wr S'_j$.
                \item $S_k \notwr S'_j$, $\forall S_k \in G(x) \setminus \{S_i\}$.
                \item $S'_j \notwr S_i$.
            \end{enumerate}

        \item[{\bf Shrinks}] This is the reverse behavior in which only a portion of a large group continues at the next control point (see Figure \ref{fig:EvolGraph1} b)). Formally, we say that group $S_i \in G(x)$ {\it Shrinks}  into group $S'_j \in G(x')$ if and only if:

            \begin{enumerate}
                \item $S'_j \wr S_i$.
                \item $S'_k \notwr S_i$, $\forall S'_k \in G(x') \setminus \{S'_j\}$.
                \item $S_i \notwr S'_j$.
            \end{enumerate}

        \item[{\bf Merges}] Two or more groups at a control point $x$ {\it Merge} into a single group at the next control point $x'$ if many of the athletes in each group in $x$ belong to the group in $x'$ (and this large group in $x'$ shares many athletes with the union of the smaller groups in $x$, see Figure \ref{fig:EvolGraph1} c)). Formally, a sequence $S_{i_1}, \cdots, S_{i_k} \in G(x)$ of groups (for $k \ge 2$) are {\it merged} into a single group $S'_j \in G(x')$ if and only if:
            \begin{enumerate}
                \item $S_{i_a} \wr S'_j$, $\forall 1\le a \le k$.
                \item $S_r \notwr S'_j$, $\forall S_r \in G(x) \setminus \{ S_{i_1}, \cdots, S_{i_k} \}$.
                \item $S'_j \wr \cup_{a=1}^{k} S_{i_a}$
            \end{enumerate}

            Note that condition (1) and Proposition \ref{Proposition:union2} imply that $\cup_{a=1}^{k} S_{i_a} \wr S'_j$. Together with condition (3) this implies
            $S'_j \sr \cup_{a=1}^{k} S_{i_a}$.

            Since the groups at a control point are disjoint, Proposition~\ref{Proposition:onlyoneset} also implies that $S'_j\notwr S_r$, $\forall S_r \in G(x) \setminus \{ S_{i_1}, \cdots, S_{i_k} \}$. Consequently, two or more groups {\it merge} into a single group if the merged group is strongly related to their union and the merged group is not weakly related to any other group.

        \item[{\bf Splits}] {\it Splits} is the reciprocal of {\it Merges}. Formally, a group $S_i \in G(x)$ {\it splits} (see Figure \ref{fig:EvolGraph1} c)) into two or more groups $S'_{j_1}, \cdots, S'_{j_k} \in G(x')$, $k \ge 2$, if and only if:

            \begin{enumerate}
                \item $S'_{j_a} \wr S_i$, $\forall 1\le a \le k$.
                \item $S'_r \notwr S_i$, $\forall S'_r \in G(x') \setminus \{ S'_{j_1}, \cdots, S'_{j_k} \}$.
                \item $S_i \wr \cup_{a=1}^{k} S'_{j_a}$
            \end{enumerate}

            As with {\it Merges}, we must have $S_i \sr \cup_{a=1}^{k} S'_{j_a}$ and $S_i\notwr S'_{\ell}$, $\forall S'_{\ell} \in G(x') \setminus \{ S'_{j_1}, \cdots, S'_{j_k} \}$.

        \item[{\bf Coheres}] This happens when a {\it merge} and an {\it expand} happen at the same time. That is, whenever groups merge but the resulting group is not related to any of the original group (or the union of all groups, see Figure \ref{fig:EvolGraph1} c)). Formally, two or more groups $S_{i_1}, \cdots, S_{i_k} \in G(x)$, $k \ge 2$, {\it cohere} into a single group $S'_j \in G(x')$, if and only if:

            \begin{enumerate}
                \item $S_{i_a} \wr S'_j$, $\forall 1\le a \le k$.
                \item $S_r \notwr S'_j$, $\forall S_r \in G(x) \setminus \{ S_{i_1}, \cdots, S_{i_k} \}$.
                \item $S'_j \notwr \cup_{a=1}^{k} S_{i_a}$.
            \end{enumerate}

            Note that the first condition and Proposition \ref{Proposition:union2} imply that $\cup_{a=1}^{k} S_{i_a} \wr S'_j$. However, in this case the two groups are not connected because of the third condition.

            From (3.) we have $I(S'_j, \cup_{a=1}^{k} S_{i_a}) < \mu$, thus $I(S'_j, S_{i_a}) < I(S'_j, \cup_{a=1}^{k} S_{i_a}) < \mu$, and consequently:

            $\forall S_{i_a}, \,\, 1\le a \le k, \,\, S'_j \notwr S_{i_a}$.

        \item[{\bf Disbands}] The inverse of {\it Cohere}. In this case a single group {\it Disbands} into smaller groups at next control point, but still there is no relationship between the union of the small groups and the original group (see Figure \ref{fig:EvolGraph1} c)). Formally, a group $S_i \in G(x)$ {\it disbands} into groups $S'_{j_1}, \cdots, S'_{j_k} \in G(x')$, $k \ge 2$ if and only if:

            \begin{enumerate}
                \item $S'_{j_a} \wr S_i$, $\forall 1\le a \le k$.
                \item $S'_r \notwr S_i$, $\forall S'_r \in G(x') \setminus \{ S'_{j_1}, \cdots, S'_{j_k} \}$.
                \item $S_i \notwr \cup_{a=1}^{k} S'_{j_a}$.
                \end{enumerate}

            As with {\it Coheres}, we have $\cup_{a=1}^{k} S'_{j_a} \wr S_i$, but this is a weak relationship.
    \end{description}

    Note that, we introduced these patterns in a way that they are mutually exclusive: the groups involved in a pattern between two control points $x < x'$ cannot participate in any other pattern at the same two control points. This will facilitate the design of Algorithm~2 for detecting the different behaviors.

\section{Evolution graphs} \label{evolutionGraphs}

    Next, we introduce the evolution graphs, we will use as a tool for analyzing the evolution patterns of groups of athletes during a race. Along the paper we will use standard concepts of graph theory. See~\cite{Harary} for more details on these concepts.

    For every pair of consecutive control points $x$ and $x'$, $x<x'$, we consider a weighted directed bipartite graph $\mathcal{B}(x,x')$, which we call the {\it evolution graph} of $x$ and $x'$, defined as follows. The two sets of vertices are the groups $G(x)$ and $G(x')$, respectively. We add a directed {\it forward edge} from vertex $S_i \in G(x)$ to vertex $S'_j \in G(x')$ if and only if $S_i \wr S'_j$. Similarly, we add a directed {\it backward edge} from vertex $S'_j \in G(x')$ to vertex $S_i \in G(x)$ if and only if $S'_j \wr S_i$. Hence, there will be both forward and backward edges to and from two vertices $S_i \in G(x)$ and $S'_j \in G(x')$ if and only if $S_i \sr S'_j$.   Because of Propositions \ref{Pro:onlyonegroup} and \ref{Pro:onlyonegroupBidirecccional}, the out-degree of any vertex in the graph $\mathcal{B}(x,x')$ is at most one, while its in-degree can be larger than one. Observe that isolated vertices can also exist. The weight associated to the edge connecting $S$ and $S'$, independently of being a forward or backward edge, is defined as $|S_i \cap S'_j|$.

    An example of an evolution graph is shown in Figure \ref{fig:EvolGraph}. The groups representing the vertices at control points $x$, $x'$  are sorted from top to bottom in the order in which they were recorded at the control points. The size of each group is proportional to the number of runners that form the group.

    \begin{figure}[!htp]
        \begin{center}
        \includegraphics[width=0.45\textwidth]{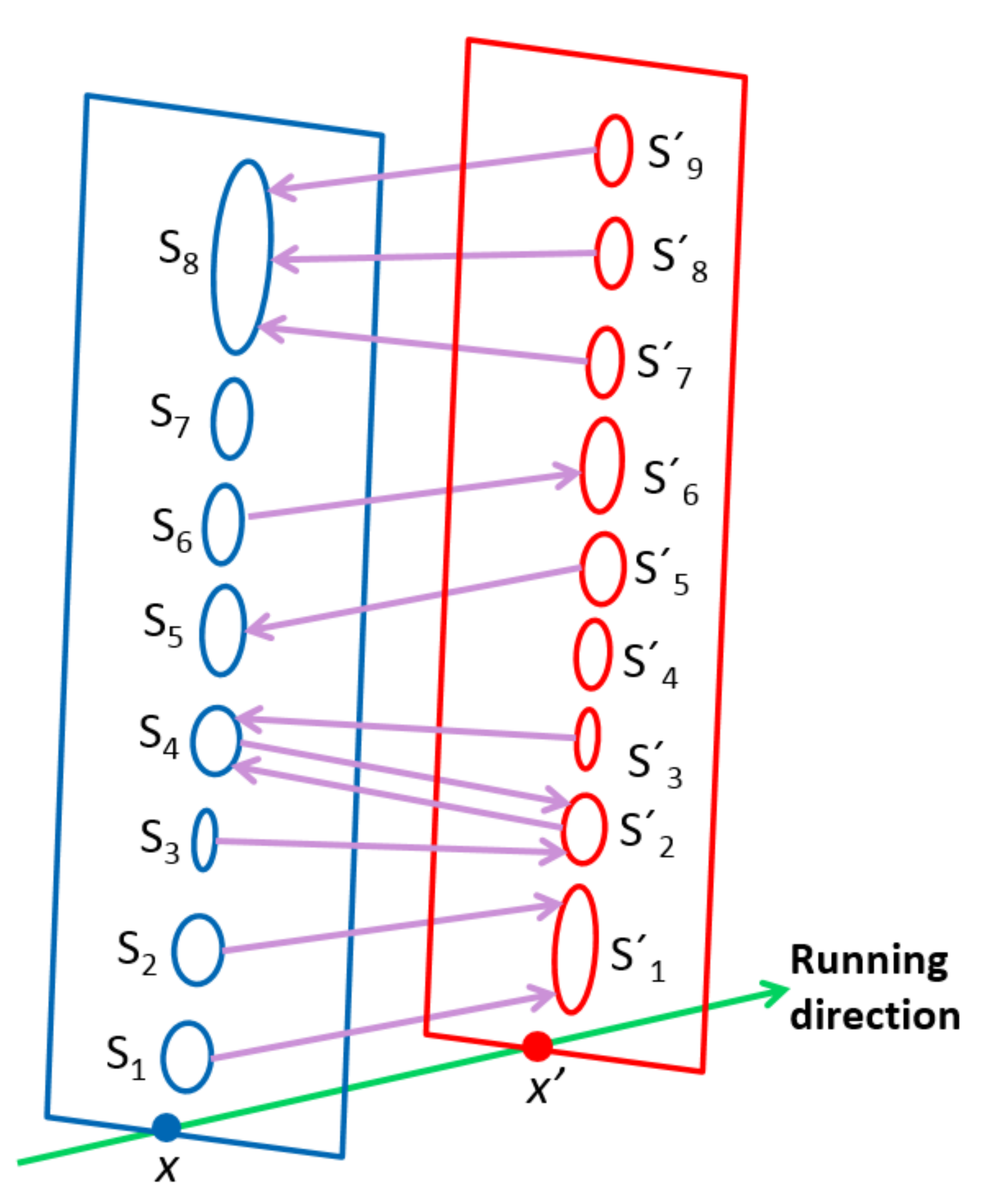}
        \end{center}
        \caption{Example of an evolution graph for two consecutive control points $x < x'$.}%
        \label{fig:EvolGraph}
    \end{figure}

    The evolution graph encapsulates the information on weak and strong relations between groups from two consecutive control points, thus it can be used to analyze the evolution patterns of these groups (see Figure \ref{fig:EvolGraph1}).

\section{Description of the algorithms} \label{algorithms}

    We now focus on the computational aspects of the group evolution patterns described in Section \ref{patterns}. First, in Section \ref{evolGraphComput}, we show how to compute the groups at each control point and the evolution graph between two consecutive control points. Then, in Section \ref{groupPatternComput} we discuss how to extract the group evolution patterns from the evolution graph.

\subsection{Group and evolution graph computation} \label{evolGraphComput}

    As is common practice in athletic races, we consider the identifiers of athletes to be known in advance as well as the positions of all control points. We consider the events (i.e. an athlete reaches a particular control point) to be sorted in terms of the time when they happen. Notice that consecutive events (in time) are likely to happen at different control points. This is given by the motivational problem, where the system that controls the race processes events in the order that they happen. Nevertheless, this could be used with generalized data by simply sorting the data by time first, thus "simulating" this kind of input.

    Our algorithm processes the events in increasing time order. For every event, the control point it corresponds to is examined. If the athlete causing the event is the first to reach the control point, a component is {\it started}. Otherwise, we check whether more than $\varepsilon$ seconds have elapsed from when the last athlete went through that control point. If not, the athlete is added to the active component, otherwise, the active component in that control point is decreed {\it finished} and a new one is started (this new one will contain the newly arrived athlete). If the number of athletes in the finished component is greater than or equal to the group threshold, the component is stored as a group and assigned a unique identifier within its control point. For every athlete $e_j$ we keep, at all times, a doubly connected list {\it groups($e_j$)} that stores the identifiers of the groups $e_j$ belonged to at the control points that the athlete has already run through (or a code value $\emptyset$ they did not belong to a group at one or more control points).

As the components are detected, a bipartite graph $\mathcal{P}(x,x')$ is built for every pair of consecutive control points $x$ and $x'$, and its nodes are the already detected groups and the active component of each control point. Its edges join two nodes whenever their groups have no empty intersection and the cardinality of the intersection weights the edge. This graph is the precursor to the evolution graph $\mathcal{B}(x,x')$ and is called the {\it precursor graph}.  Basically, every time a component is decreed finished, the precursor graph is updated. If the component does not define a group, its corresponding node, together with the edges it may have, are deleted from the precursor graph. When a group $S'_j \in G(x')$ is detected, the (undirected) edges between $S'_j$ and all the groups in $G(x)$ with no empty intersection with $S'_j$ are added in $\mathcal{P}(x,x')$. In this case, once the precursor graph has been updated, then the evolution graph is also updated recording the existent (forwards/backwards) relations between the new group and the existing groups.  For further details see Algorithm~1 and the description below.

{\small
    \algfigure{
    \\
    {\sc Algorithm 1: \bf  {Compute Groups and Evolution Graphs} } \\
    \> \For {all events $(x_j,t_{j,i})$ in input (athlete $i$ runs through control point $x_j$ at time $t_{j,i}$)  }\\
    \> \> {\em \{Let $x=x_{j-1}$, $x'=x_j$ and $x''=x_{j+1}$ \} }\\
    \> \> {\em \{Consider $\mathcal{C}$, the active component of the control point $x'$ \} }\\
    \> \> \If { $\mathcal{C}$ does not exist } \\
    \> \> \> {\em \{Create $\mathcal{C}$, the active component, with list: $l=\{e_{i}\}$ \} } \\
    \> \> \Else \\
        \> \> \> {\em \{Retrieve the last event in $\mathcal{C}$ $\rightarrow$ $(x',t_{j,k})$ \} } \\
        \> \> \> \If { $|t_{j,i}-t_{j,k}|<\varepsilon$} \\
        \> \> \> \> \{(same component) $\rightarrow$  Add athlete $e_i$ to $l$ the list of athletes of $\mathcal{C}$\} \\
        \> \> \> \Else \{Finish the component $\mathcal{C}$ \}  \\
        \> \> \> \> {\em \{Retrieve the list of athletes in $\mathcal{C}$: $l=\{e_{i_1},...e_{i_r}\}$ \} } \\
        \> \> \> \> \If { $|\mathcal{C}|<Group\_Threshold$} \{ This component was not a group  \} \\
        \> \> \> \> \> \Foreach {$e_{i_r} \in l$  } Do \\
        \> \> \> \> \> \> groups($e_{i_r}$) $\leftarrow$ $\emptyset$ \{ For every athlete push back new $\emptyset$ information \}\\
        \> \> \> \> \> {\em \{ \bf {Call Function:} $Delete\_Tentative\_Edges(\mathcal{P}(x_{j},x_{j+1}),\mathcal{C})$ \} } \\
        \> \> \> \> \Else \{New group $S'$ detected $\rightarrow$ Assign unique group identifier $Id$ \} \\
        \> \> \> \> \> \Foreach {$e_{i_r} \in l$  } Do \\
        \> \> \> \> \> \> groups($e_{i_r}$) $\leftarrow$ $Id$  \{ For every athlete push back new group information \}\\
        \> \> \> \> \> {\em \{ \bf {Call Function:} $Update\_Precursor\_Graph(\mathcal{P}(x,x'),l)$ \} } \\
        \> \> \> \> \> {\em \{ \bf {Call Function:} $Update\_Evolution\_Graph(\mathcal{B}(x,x'),\mathcal{B}(x',x''),\mathcal{P}(x,x'), \mathcal{P}(x',x''),\mathcal{C})$ \} } \\
        \> \> \> \> {\em \{Create a new active component $\mathcal{C}$ with list: $l=\{e_{i}\}$ \} } \\
    \\
    }
}

    The function $Delete\_Tentative\_Edges$ is called when the active component $\mathcal{C}$ of $x'$ is decreed finished but it does not have enough entities to be considered a group. Since $\mathcal{C}$ is not a group it will not be stored and the component itself and the tentative edges incident to it are deleted from $\mathcal{P}(x',x'')$.

    The function $Update\_Precursor\_Graph$ basically considers all the athletes in the newly appearing group $S'$ at the control point $x'$ (stored in a list $l$), it looks at the group they belonged to back at the previous control point $x$ and updates the precursor graph $\mathcal{P}(x,x')$. This is done by adding an edge of weight one between the two groups if the edge does not exist, or incrementing its weight in one, otherwise. The athletes in $l$ that still have no group information for control point $x$ are athletes contained in the active (not yet finished) component of $x$ and thus they count in the tentative existent edge of $\mathcal{P}(x,x')$. Consequently, the function traverses the list of athletes $l$ and accumulates the number of entities shared between the newly appearing group and every group in $x$ that shares at least one athlete with it.

    Finally, the function $Update\_Evolution\_Graph$ is called when $\mathcal{C}$ has become a group $S'_j\in G(x')$. For each edge of $\mathcal{P}(x,x')$ and $\mathcal{P}(x',x'')$ between groups $S'_j$ and $S_i \in G(x)\cup G(x'')$  whether $S'_j \wr S_i$ or $S_i \wr S'_j$ is checked. If a relation holds, the corresponding forward or backward edge is added in the evolution graph.

\subsubsection{Complexity}

Let $n$ denote the number of athletes and $\delta$ the number of control points. Let us also keep in mind that the group threshold $m$ is a small constant value. Since each athlete passes through each control point at most once, the total number of events is $O(n\delta)$. The key to having a fast algorithm is that most operations introduced per event take constant time. Additionally, the few operations that need more than constant time per runner are amortized over the runners.

When runner $i$ passes through the $j$-th control point event $(x_j,t_{j,i})$ is generated. At this point, we retrieve the current active component $\mathcal{C}$ and determine the last already processed event $(x_j,t_{j,k})$ of that control point. This can be done in $O(1)$ time with a simple pointer to the last event. If the time difference between the two events is smaller than $\varepsilon$ then we simply add the new runner to the active component.

   A more interesting case happens when $t_{j,i}$ and $t_{j,k}$ more than $\varepsilon$ units apart. Whenever this happens, the current component $\mathcal{C}$ associated to runners $\{e_{1},...e_{r}\}$ is completed. We create a new component $\mathcal{C}'$ whose (for now) only runner is $i$. In addition, we check is the size of $\mathcal{C}$: if it contains less runners than the group threshold $m$, then it does not define a group. In such a case we note down that none of the runners associated to $\mathcal{C}$ belong to a group and we are done. Otherwise, we have found a new group and must find any possible relationship with previously existing groups.

In order to determine if two groups are related we need the cardinality of their intersection. Thus, for each runner in $\mathcal{C}$, we check to what group (if to any at all) it belonged to in the previous control point. This information can be accessed in constant time with the runner's id. By adding this information over all runners $\{e_{1},...e_{r}\}$ we can determine (in $O(r)$ time) the cardinality of the intersection of $\mathcal{C}$ with each of the groups of the previous control point. Once the cardinality of each of these intersections is known, we can determine any weak relationship $\mathcal{C}$ may have (and add those edges to $\mathcal{P}(x,x')$).

In essence, a runner is accessed at most twice for each control point during the processing of events. First when the runner passes through the control point and second one when the component to which the runner belongs to is completed. In both events the runner only causes a constant number of operations, thus we conclude that the total time spent in the whole execution along all control points is $O(n\delta)$.%

\subsection{Computing group evolution patterns}  \label{groupPatternComput}

    The group evolution patterns are computed with the help of the evolution graph $\mathcal{B}(x,x')$. With the graph in hand, we can find the patterns by querying the evolution graph $\mathcal{B}(x,x')$ and check when are the conditions of each pattern met. We can report patterns on-the-fly (analyze $\mathcal{B}(x,x')$ after we detect a new group $S'\in G(x')$ or $S\in G(x)$), or once all events have been processed (i.e., when the race is over).

    Reporting the patterns as soon  as they are detected implies reporting transitory patterns, some of which change over time. %
       This is not a limitation of the algorithm, but rather a result of having only partial information available until all the athletes who were originally part of the group have run through the next control point. This partial information may cause a {\it Split} pattern to report several {\it Appear} patterns before detecting the {\it Split} itself. For instance, if group $S$ splits into $S'_{i_1}, \ldots, S'_{i_k}$, and several groups $S'_{i_a}\in G(x')$ finish before $S$ is finished, every such $S'_i$ would be detected as  {\it Appears} until $S$ is finished. Something similar may happen with {\it Shrinks} and {\it Survives}, {\it Disbands} and {\it Splits} or {\it Coheres} and {\it Merges}, among others. To detect only real patterns one should first wait until $\mathcal{B}(x,x')$ is completed and all the information is available.

The queries needed to detect the patterns follow directly from their definition and are formulated mainly by using the forward-in-degree of the groups $S'\in G(x')$, $f_{in}(S')$, and the backward-in-degree of $S \in G(x)$, $b_{in}(S)$, of $\mathcal{B}(x,x')$. Sometimes the forward-out and backward-out degrees of the respective groups are also needed $f_{out}(S)$ and $b_{out}(S')$ as it can be seen in Algorithm~2. Algorithm~2 has as input $S'\in G(x')$ and $\mathcal{B}(x,x')$, another equivalent algorithm with input $S\in G(x)$ and $\mathcal{B}(x,x')$ is needed to detect the patterns on-the-fly, but, because it is very similar to Algorithm~2, is not presented here.

{\small
      \algfigure{
        \\
        {\sc Algorithm 2: \bf  {Detect patterns} }\\
        \> INPUT:  $S' \in G(x')$, $\mathcal{B}(x,x')$\\
        \> \> \If $f_{in}(S') =0$  \quad \{ $\nexists S \in G(x) \; | \; S \wr S' $\} \\
        \> \> \> \If $b_{out}(S') = 0$ \, \quad \{ $\nexists S \in G(x) \; | \; S'\wr S$ \} \, $\rightarrow$ Report \it{Appears}   \\
        \> \> \> \Else $\{ \, b_{out}(S') = 1$,  $\exists S\in G(x) \, | \, S' \wr S\, \}$\\
        \> \> \> \> \If $b_{in}(S) = 1$ \, $\rightarrow$  Report {\it Shrinks}\\
        \> \> \> \> \Else  \{ $b_{in}(S) = k>1$, consider $S'_{i_a} \in G(x')$ with $1\le a \le k \, | \, S'_{i_a} \wr S$ \} \}\\
        \> \> \> \> \> \If   $S \wr \cup_{a=1}^{k} S'_{i_a}$ $\rightarrow$    Report {\it Splits}\\
        \> \> \> \> \> \Else   $\rightarrow$   Report {\it Disbands}\\
        \> \> \Else \If $b_{out}(S')=1$ and $S\wr S'$ \{ where  $S\in G(x)$ holds $| \, S'\wr S$ \} $\rightarrow$ Report {\it Survives} \\
        \> \> \> \If $f_{in}(S')>1$  $\rightarrow$ Report {\it Absorbs} \\
        \> \> \> \Else \If $b_{in}(S)>1$  $\rightarrow$ Report {\it Spawns} \\
        \> \> \Else  \If { $f_{in}(S')=1$  } \, \{$\exists ! \, S \in G(x) \, | \, S \wr S' $ \} \, $\rightarrow$ Report \it{Expands}  \\
        \> \> \Else \{ $f_{in}(S') = k>1$, consider $S_{i_a} \in G(x)$ with $1\le a \le k \, | \, S_{i_a} \wr S'$ \}\\
        \> \> \> \If { $S' \wr \cup_{a=1}^{k} S_{i_a}$ }  $\rightarrow$ Report \it{Merges}  \\
        \> \> \> \Else $\rightarrow$ Report \it{Coheres}\\
        \\
    }\label{alg2}
}
    Algorithm~2 does not detect {\it Disappear} patterns because they can only be detected by analyzing the groups $S\in G(x)$ which are, in fact, the groups that disappear. {\it Disappears} is equivalent to the {\it Appears} pattern but when considering $S\in G(x)$ instead of $S'\in G(x')$, i.e. a group $S\in G(x)$ {\it disappears} whenever $f_{out}(S)=0$ and $b_{in}(S)=0$.\\

    If we detect patterns on-the-fly, after detecting a new group $S'\in G(x')$ we have to analyze the updated edges of the evolution graph $\mathcal{B}(x,x')$ using Algorithm~2, and $\mathcal{B}(x',x'')$ with $x''>x'$, the next control point, by using the analog algorithm. On the other hand, if this has been done once the evolution graph $\mathcal{B}(x,x')$ is completed, we can use Algorithm~2 to detect all the patterns (except for {\it Disappears}) and analyze the groups $S\in G(x)$ by only checking the {\it Disappears} pattern. To detect that $\mathcal{B}(x,x')$ is complete, we need to know that all the athletes who have not left the race have run through control points $x$ and $x'$. This can be achieved by using a "broom wagon". In fact, detecting if all the athletes have run through a control point is also needed to decree the last component of each control point finished.

\subsubsection{Complexity}

    From Algorithm~2, we conclude that detecting the patterns for {\it Appears} and {\it Disappears} requires simple in-degree/out-degree checks that take $O(1)$ time. The patterns for {\it Survives}, {\it Expands} and {\it Shrinks} require checking the cardinality of the intersection of the two groups involved in one single relation. Recall that the size of the intersection is computed when the second group is completed and is stored in the edge between them (i.e. for $S\in G(x)$ and $S'\in G'(x)$ such that $S\wr S'$ or $S'\wr S$ we store $|S\cap S'|$ in the edge). Thus, to determine if any of these events occurs, it suffices to do a constant number of checks. Detecting {\it Splits}, {\it Disbands}, {\it Merges} and {\it Coheres} is similar. The only difference is that we need to take into account all incoming or outgoing edges (and cardinalities of intersections) adjacent in $\mathcal{B}(x,x')$ to a group (i.e. for $S\in G(x)$ we store $|S\cap_{a=1}^{k}S'_{i_a}|$ whenever $S\wr S'_{i_a}$, and for $S'\in G(x')$ the value $|S'\cap_{a=1}^{k}S_{i_a}|$ whenever $S'\wr S_{i_a}$).

  Detecting the patterns appearing in $\mathcal{B}(x,x')$ takes $O(1)$ time per group in $G(x) \cup G(x')$. As all the information that we need to check is available in constant time, the total time we spend looking for all of the above patterns involving one particular group is proportional to its degree. Since the total degree in the evolution graph is bounded by the number of events as there exist at most $\frac{2n}{m}$ groups, determining the group evolution patterns appearing between two time steps take $O(\frac{n}{m})$, and all the evolution patterns $O(\frac{n\delta}{m})$. Bearing in mind that the group threshold $m$ is a small constant value we conclude that these calculations do not increase the time complexity of the algorithm to obtain the evolution graph. Thus, all the groups, evolution graphs and group evolution patterns that appear when the race is analyzed, can be obtained in $O(n\delta)$ time.

\section{Long-term patterns}\label{problems}
    The evolution graphs allow us to characterize and compute the evolution of the groups between two consecutive control points, and detect the evolution patterns of the existing groups. In this section, we aim to acquire and elaborate the information about the groups affecting a larger number of consecutive control points.

    To combine the information extracted from each weighted directed bipartite evolution graph $\mathcal{B}(x_{i},x_{i+1})$ with $1\le i < \delta$, we consider a graph that is obtained as the concatenation of all such graphs. Specifically, this directed graph is defined as $\mathcal{R}=(\mathcal{V},\mathcal{E})$ with   $\mathcal{V} = \bigcup_{i \in \{1\cdots \delta \}} \bigcup_{S \in G(x_i) } S$  and  $\mathcal{E}$ defined as the union of all the edges in the directed bipartite evolution graphs $\mathcal{B}(x_{i},x_{i+1})$. Consequently, the vertices in $\mathcal{R}$ correspond to the groups at all control points, and the edges between them, and the (directed) weak relations between the groups. We call this graph the {\it global graph}.

    The directed graph $\mathcal{R}$ allows us to introduce some problems related to the long-term patterns of the groups over several consecutive control points. We define four different long-term group patterns: surviving, traceable forward, traceable backward and related forward or backward. We are interested in determining the {\it length}, i.e. the number of consecutive control points where a specific long-term group  pattern appears - and also the largest length for which each pattern occurs. We would like to remark that equivalent problems associated to any of the (short-term) evolution patterns could be considered.

    \begin{enumerate}
        \item {\it Surviving}.  A group survives between two consecutive time-steps (as groups $S\in G({x_i})$, $S'\in G({x_{i+1}})$ if both a forward edge from $S$ to $S'$ and a backward edge from $S'$ to $S$ are present and, consequently, a strong relation is detected in the graph. Informally, we can say that for each group, $S\in \mathcal{V}$, we are interested in determining the number of consecutive control points along which that group survives. Formally it stands for finding the length of the longest path that can be traversed in $\mathcal{R}$ using both forward and backward edges that contain $S$. In graph theoretical terms, we are looking for the length of the longest {\it strongly connected component} of $\mathcal{R}$ containing $S$. %

        \item {\it Traceable forward}. We define the length of the traceable forward pattern of a group as the number of consecutive control points that a group either {\it Survives, Expands, Merges} or {\it Coheres}. In graph theoretical terms, we look for the length of the longest path in $\mathcal{R}$ that contains a group $S$ and that uses only forward edges.

            This allows us to detect behaviors such as a large group staying together for a long period and towards the end splitting into two smaller groups, or one of the small groups disappearing, but the other continuing for the rest of the race. The surviving long-term pattern of such a group would end when the large group splits into two, but it would make sense to say that a long-term pattern lasted the whole race thanks to the small group that has lasted the whole race. Hence it makes sense to look for the length of the traceable forward relation.

          \item {\it Traceable backward}. This is equivalent to the traceable forward long-term pattern, but allows the groups connected via {\it Shrinks, Splits} or {\it Disbands} to be detected. This pattern analyzes the length of the longest paths in $\mathcal{R}$ that use only backward edges.

        \item {\it Related forward or backward}. As a combination of the traceable forward and the backward relations, we now look for paths of groups that have maintained some relationship for the greatest number of control points during the race. This connection can be either a weak or strong relation, regardless of the direction. In graph theoretical terms, we virtually transform $\mathcal{R}$ into an undirected graph and look for the length of the largest connected component containing a vertex.
    \end{enumerate}

\subsection{Long-term pattern computation}

    The length of the pattern is computed by making a sweep of the global graph. For each group we use four helpful integers $lpS$, $lpF$, $lpB$, and $lpR$. Each of these integers is associated to one of the four long-term patterns described in Section~\ref{problems}. Initially, the four variables are set to zero for all groups.

    For simplicity, we compute long-term patterns once the race has finished and thus no changes happen on the global graph. However, we emphasize that this algorithm can be adapted to a dynamic situation in which data comes in real-time. It suffices to update the values of these four variables as the edges are added one by one.

    We start by describing how to compute the length of the surviving relation. Intuitively speaking after the sweep, the variable $lpS$ of a group $S$ will denote how long of a sequence of surviving edges can we follow until we reach $S$. For groups at the first control point, there is no path we can follow to reach them, so their variable $lpS$ will be zero.

    The sweep scans through all groups in the global graph, starting from the groups in the first control point, then those in the second control point and so on. When scanning a group $S \in G(x_i)$, we first check whether it has a strong relation with some group $S' \in G(x_{i+1})$ (recall that, by definition of strong relationship, there can be only one such group). If this happens, we set the value of variable $lpS$ of $S'$ as one higher to the value of the same variable in $S$. Thus, if we have a chain $S_1, S_2, S_3$ and $S_4$ of four groups that are strongly related to each other, the values of their $lpS$ will be $0,1,2$ and $3$, respectively. Observe that the longest surviving group can be traced by finding the group whose associated $lpS$ is largest (and increasing it by one).

    The length of the other three relations is computed in a similar way by sweeping along the groups in the global graph and transmitting the information along the paths. In the following we only describe the differences. The major change when computing the length of the traceable forward pattern is that whenever $S \in G(x_i)$ has a forward relation with some group $S' \in G(x_{i+1})$ we set the value of variable $lpF$ in $S'$ as the maximum of either its current value or the value of $lpF$ in $S$ plus one. The reason for doing so is that there can be several groups that are weakly related with $S'$. Since we are interested in the longest pattern possible, we would select the largest of the two.

    When computing the backwards relation we are interested in orientations in reverse order. As a result, we also sweep the groups in reverse order (that is, we start with the last control point and sweep towards the first control point). The values at this sweep are stored in $lpB$.

   Finally, when computing the length of related forward and backwards at the same time we can sweep in either direction (say, forward). However, in this case group $S\in G(x_i)$ must spread the information to all groups $S'\in G(x_{i})$ such that either $S\wr S'$ or $S'\wr S$. In either of the two cases, we may increase the value of $lpR$ of $S'$ to the value of the same variable of $S$ plus one (if it is a larger value).

   Figure \ref{fig:LongTerm} presents one example of teh

    \begin{figure}[!htp]
        \begin{center}
        \includegraphics[width=0.45\textwidth]{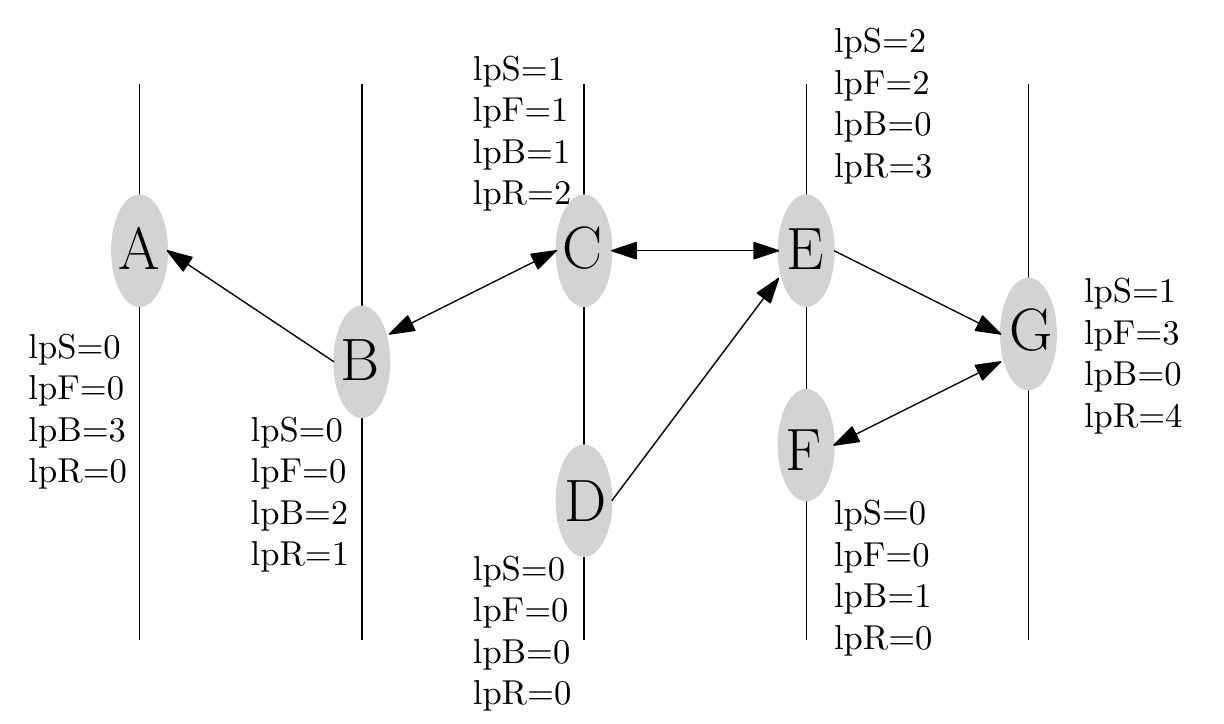}
        \end{center}
        \caption{Example of Long-term pattern computation. The long term patterns in this case are: Surviving: BCE, Traceable forward: BCEG, Traceable backwards: ECBA, Related forward or backward: ABCEG}%
        \label{fig:LongTerm}
    \end{figure}

\subsubsection{Complexity}

    When computing any of the four long term patterns, we do a single sweep of the global graph. Each time we check a group we look for its adjacencies  (forward, backwards, or both depending on the type of pattern we are looking at). Each edge of the graph generates a constant number of operations. Since the graph has all adjacencies available at constant time from each node, the overall time spent in any of the sweeps is proportional to the number of vertices and edges of the graph.

    Remember that each group must contain at least $m$ athletes and that there is one vertex in the global graph per group. Thus, we directly conclude that there are at most $n\delta/m$ vertices in the global graph. Moreover, each vertex in the global graph has at most two outgoing edges (one each to a group in the the previous and following control points), hence the number of edges is also $O(n\delta/m)$. Since both the number of vertices and edges of the global graph is linear in the number of events and the group threshold $m$ is a constant, we conclude that the long term patterns can be computed in linear ($O(n\delta/m)$) time.

\subsection{Longest behavior detection}

    Once the sweeps are finished, computing the largest length for which each long-term pattern appears is immediate. As mentioned before, we merely have to find the group with the largest value $lpS$, $lpF$, $lpB$, and $lpR$, respectively. We keep track of the largest value of each of the four traits (and a reference of the related group) with four global variables. The other groups that are part of the pattern can be obtained by following adjacencies along the graph.

            Note that, for simplicity in the description (and analysis), we described the computation of the four long term patterns separately. However, from a practical point of view it is simpler if we do it in two sweeps: a forward sweep in which $lpS$, $lpF$ and $lpR$ variables are set, and a backwards one to compute $lpB$. Modifications needed to do so are straightforward.

In this paper we focused on the technical difficulties of extracting group information. A possible future work direction would be to explore good methods of displaying this information once it has been obtained. Cognitive experts agree that graphs drawn without crossings help ease the understanding of the information being displayed~\cite{DBLP:journals/vlc/HuangEH14}. Because of the nature of the evolution and the global graphs, it is very likely that the graphs we defined are planar (that is, that we can draw them in a paper without any edges crossing with other edges). In some rare cases, the interaction between groups may force crossings, but even if this happens the graphs will remain {\it quasi-planar}. There is a large amount of literature on how to characterize and visualize graphs that fit in this category and how to display them in a way that the few crossings do not deter the understanding of the information~\cite{JGAA-459,shonan,hong_et_al,DBLP:conf/ictcs/Liotta14}. Because we have edges fanning out of a vertex (in case a group splits or disbands) or merging in (in the reverse operations), the family of {\it k-gap-planar-graphs} seems like a perfect fit. We refer the interested reader to~\cite{gap} for more information on this kind of graphs.

\section{Experimental results} \label{experiments}

    In this section, we provide experiments to study in depth the algorithms presented. All algorithms were implemented in C++ and the experiments where run using one single 2.4GHz processor with a  and 128G of RAM memory running under a Linux Ubuntu operating system. The computer used had multiprocessor capacity (12 dual core processors) but the code was written to be run in a single processor and no concurrent executions (several processes using the same code) where not used. In order to measure the time needed by each part of the algorithm we used the {\it Callgrind} profiling tool of the Valgrind instrumentation framework \cite{Valgrind}.

    Concerning the data used, Even though a large number of runners (say, 30,000) existing real data is sparse (typically only 12 control points are placed along the length of a full marathon). Consequently, in addition to considering real data we also include experiments with denser synthetic datasets to provide an indication of how the algorithms presented could be used in the near future when denser real data becomes widely available. Specifically, in Section \ref{syntheticEXP} we present experiments with synthetic data, which allows us to have absolute control over the conditions of the experiment and to test our implementation in very demanding conditions. At the same time, we can test the correctness of the implementation. In Section \ref{realEXP} we test and experimentally analyze the algorithms with the denser real data we have been able to obtain. There we analyze the impact of the parameters $m$, $\mu$ and $\varepsilon$ in the detected patterns. In both experiments we analyzed the obtained results, provide the running time of the whole algorithm and compute the in-average percentage of time needed by each step of the algorithm. The obtained results corroborate the theoretical complexity analysis and show how our algorithms can be used to extract meaningful information from the data.

    Additionally, we are able to maintain for every athlete: 1) their position in the race, 2) their current pace (i.e. minutes needed to run a km) and average pace over the race, 3) the groups they belonged to at every past control point and all relative group patterns. All the information is accurate up to the last control point the athlete went through. Since during the data processing the information associated to an athlete is accurate up to the last control point the athlete went through, many other interesting information immediately follows and can be easily reported. For race directors, it is possible to determine when an athlete skips a control point or detect significant changes in pace (indicating possible cheating, problems with equipment, or even athlete exertion problems). As our algorithms provide information on currently existing groups and their behavior, more complex issues can be detected and even anticipated. An example of this would be for traffic problems that might occur when a large number of runners approaches a section of the course with lower capacity (as happens, for example, in races that use a single street in both directions and have a point where all runners need to turn 180 degrees). On the other hand, every runner could easily be informed, via a phone or smartwatch app, not only on their current race standing, but also on the pace of nearby competitors compared to their own (averaged in terms of their group) or a comparison with other groups. It would, therefore, be possible to know whether the pace of the runner is increasing/decreasing compared to the average of runners in nearby groups and even to obtain a projection of the number of runners that might overtake/be overtaken by the runner. While it should be acknowledged that there is a potential for an excess of information, because, after all, for most people races are more about recreation than competition; this information is clearly of interest for sports broadcasts or for the more competitive runners and their coaching staff.

\subsection{Experiments with synthetic data} \label{syntheticEXP}

   In order to test the performance of our algorithms in a demanding scenario, while at the same time being able to test the correctness of our implementation we built several synthetic data sets in which we could control and exactly know the number of occurrences of each pattern. The basic unit that we used for the test was the group of runners. Essentially, each group was created around a fixed pace. This was inspired by the groups of pacesetters used in real races. In this basic behavior, runners randomly deviate from their pace slightly, but never lose $\varepsilon$-connectedness. Whenever this {\it constant running behavior} is kept between two control points, a {\it Survives} evolution pattern is created. If this goes on for more control points, a surviving long-term pattern is also produced.

    To simulate the rest of the evolution and behavior patterns, the following events were designed. Thus, a group (or part of it) at any given control point can remain either unchanged or:

    \begin{itemize}
        \item {\it Divide}: A group at a given control point gets broken down into several sub-parts that are not $\varepsilon$-connected. The number of subsets created was adjusted to fit the needs of the behavior that was being simulated.

        \item {\it Explode}: In this case, the entities involved get separated $\varepsilon+1$ seconds from one another so the group or part of a group they belonged to no longer exists at the control point.
    \end{itemize}

    By combining the {\it constant running}, {\it dividing} and {\it exploding} behavior, all the group evolution patterns described throughout this paper can be simulated. For example, a group dividing into two or more groups produced a {\it Splits} evolution pattern (and a later {\it Merges} evolution pattern if it went back to {\it constant running}). A group breaking up into two parts, one of which explodes, produces a {\it Shrinks} evolution pattern, or a group that has divided into seven parts, five of which explode, and then goes back to {\it constant running} produces a {\it Coheres} evolution pattern.

    For this experiment we set $m$, the group threshold, to be 7. To build our synthetic test scenarios we used aggregations of 25 runners which we will refer to as a {\it pack} (not to be confused with the concept of {\it group}). We then randomly assigned a pace to each pack. Ten paces were considered in the first test bench and 50 in the second. Consequently, many packs coexisted within the same pace. Each of the packs was also randomly assigned a behavior at each control point. Consequently, the packs at each control point formed groups and the entities that were shared by the groups from one control point to the next could be determined by checking the packs that had been assigned the same pace.

    This construction allowed us to predict (and check) the number of group evolution patterns of each type that would happen, ensuring the correctness of our implementation. Similarly, it was also possible for us to predict the maximum length of all behavior types. Additionally, we could also produce scenarios with complex group interactions that illustrate the computation capacities of our algorithm.

    Table \ref{taulaSintetics} shows the results of the experiment considered over a marathon distance course. The table includes the total number of athletes considered, the control points used, the total time taken by the algorithm, the average time needed to process an event (each event is an athlete running through a control point) and the number of times some of the evolution patterns described were identified. As an artifact of our experiment, some of the evolution patterns appeared paired. For example, {\it Appears} and {\it Disappears}, {\it Merges} and {\it Splits}, {\it Expands} and {\it Shrinks} or {\it Disbands} and {\it Coheres} presented similar numbers of occurrences. Consequently, and for the sake of brevity, we only present one evolution pattern within each pair.

    \begin{table*}[t!]
         \caption{Details of the synthetic experiment. At the top of the table the athletes were given one of 10 possible paces, while at the bottom up to 50 different paces are considered. All times are given in seconds. Recall that $lpS$ stands for longest surviving path (see Section~\ref{problems}).}
         \centering
         \resizebox{\textwidth}{!}{
         \begin{tabular}{|c|c|c|c|c|c|c|c|c|c|c|}
         	\hline
            $|athletes|$	& $|control points|$ &	control point dist. &	total time & time/event & $lpS$	& $|Survives|$ &	 $|Disappears|$ & $|Expands|$ &	$|Splits|$	&	$|Coheres|$	\\
         	\hline
            2500000 & 100 & 421.95 & 242 & $1\times 10^{-8}$ & 100 & 4946 & 1899 & 0 & 4 & 0 \\
            250000 & 500 & 84.39 & 156 & $6.24\times 10^{-8}$ & 500 & 24946 & 9257 & 12 & 4 & 2 \\
            12500 & 1000 & 42.195 & 15 & $1.2\times 10^{-6}$ & 952 & 49042 & 1817 & 579 & 67 & 704 \\
            125000 & 5000 & 8.439 & 4414 & $7.06 \times 10^{-6}$ & 5000 & 250000 & 53814 & 1612 & 4 & 0 \\
            25000 & 10000 & 4.2195 & 3070 & $1.23\times 10^{-5}$ & 10000 & 505372 & 35777 & 7893 & 62 & 857 \\
             	\hline
            2500000 & 100 & 421.95 & 252 & $1\times 10^{-8}$ & 1 & 0 & 19361 & 23 & 934 & 0 \\
            125000 & 500 & 84.39 & 75 & $1.2 \times 10^{-6}$ & 3 & 34 & 98200 & 154 & 4740 & 0 \\
            250000 & 1000 & 42.195 & 463 & $1.85\times 10^{-5}$ & 2 & 3 & 199724 & 266 & 9687 & 0 \\
            12500 & 5000 & 8.439 & 373 & $5.97 \times 10^{-5}$ & 9 & 6341 & 683871 & 4100 & 10310 & 25482 \\
            25000 & 10000 & 4.2195 & 3726 & $1.49\times 10^{-4}$ & 6 & 8091 & 1725033 & 5770 & 72949 & 8863 \\
        	\hline
         \end{tabular}\label{taulaSintetics}}
     \end{table*}

    In the second scenario (with a larger range of different paces), we test our algorithms in much more demanding conditions. Recall that currently data available for races are sparse (for example, to date only 12 control points are used in the Boston Marathon). A typical running GPS enabled running watch might measure the position 200 times over the course of a marathon, whereas the most precise experimental GPS~\cite{WW04} takes measurements every second. Using the speed of the current official world record for a marathon~\cite{Marathon} (5.72 $m/s$, or roughly 2 hours to complete the race) as a reference, we decided to set control point distances to roughly range from 5 to 500 meters. Similarly, the most popular running races nowadays range from about 30 000 runners for the most popular marathon to around 200 000 runners in some shorter distance events. Consequently, we set our test to go from 12 500 runners to a (likely unfeasible) 2 500 000 "synthetic" athletes.

    The table shows how, even in the slowest of cases, all the computations involved took about an hour and a half, which is less than the time even the faster runners take to complete a race. Furthermore, taking into account the total time needed to process an event, we see how our algorithms run in real-time even in the most demanding conditions.

    Concerning the two scenarios considered in terms of evolution pattern complexity (10 or 50 pace possibilities), scenarios with 10 pace bands allowed for less separation between the different packs and produced more compact groups with a higher presence of the {\it Expands/Shrinks, Appears/Disappears} and, (mostly) {\it Survives} evolution patterns (as shown in the top half of the table, columns 7-9). The pre-eminence of these patterns that describe at most one to one (weak or strong) relation between groups is also reflected in the fact that, for all but one of the examples presented, there are groups  that survive together for the whole race (long-term pattern depicted in the sixth column).

    Conversely, when 50 paces were considered, there were less coincidence between the divided or exploded parts of the different packs and, thus, evolution patterns that consider weak relations between one group at one control point and at least two groups at the other were more frequently observed. This is depicted best by the final two columns in the bottom half of the table. The previous asymptotic cost analysis notwithstanding, this type of behavior requires the algorithms to go over more detailed case distinctions when evaluating the evolution pattern of each group and, thus, these scenarios produce the slower times per event. Specifically, the $1.49\times 10^{-4}$ seconds/event value shown at the bottommost cell in the fifth column is the most illustrative example. Nevertheless, this value shows how even in this tailor-made, extremely demanding situation, our algorithms can process events in the order of $10^{4}$ per second, which we feel demonstrates their potential use in practical applications.

\paragraph{Runtime analysis}

We provide a detailed account on the runtime performance of our algorithms for the case of synthetic data depicted in Figure \ref{tab:RuntimeSint}. For each execution we compute the percentage of time needed at each stage and then computed the average over all the percentages of the different executions Furthermore, with the purpose of simulating the data available in real-life situations, we generated events and sorted them so they would be passed to the algorithms in the order they would occur in race time. While the time needed by these operations was part of the execution of the algorithms (and, as such, is included in the values reported in table \ref{taulaSintetics}), we feel that these costs are not really part of the discussion on algorithm performance. In order to both tell the whole story about how the algorithms perform while making the discussion on algorithm performance as clear as possible, we first (figure \ref{tab:RuntimeSint}, left) provide times for high level steps of the program execution (including sorting) and then provide details only for the parts that we consider to properly illustrate the performance of the algorithms presented in this paper (figure \ref{tab:RuntimeSint}, right).

\begin{figure}[h]
    \begin{center}
    \includegraphics[width=0.35\textwidth]{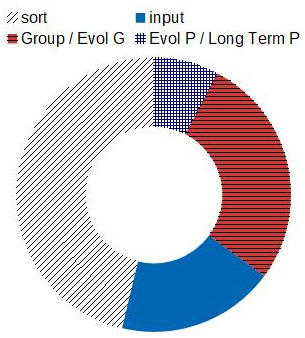} \qquad
    \includegraphics[width=0.45\textwidth]{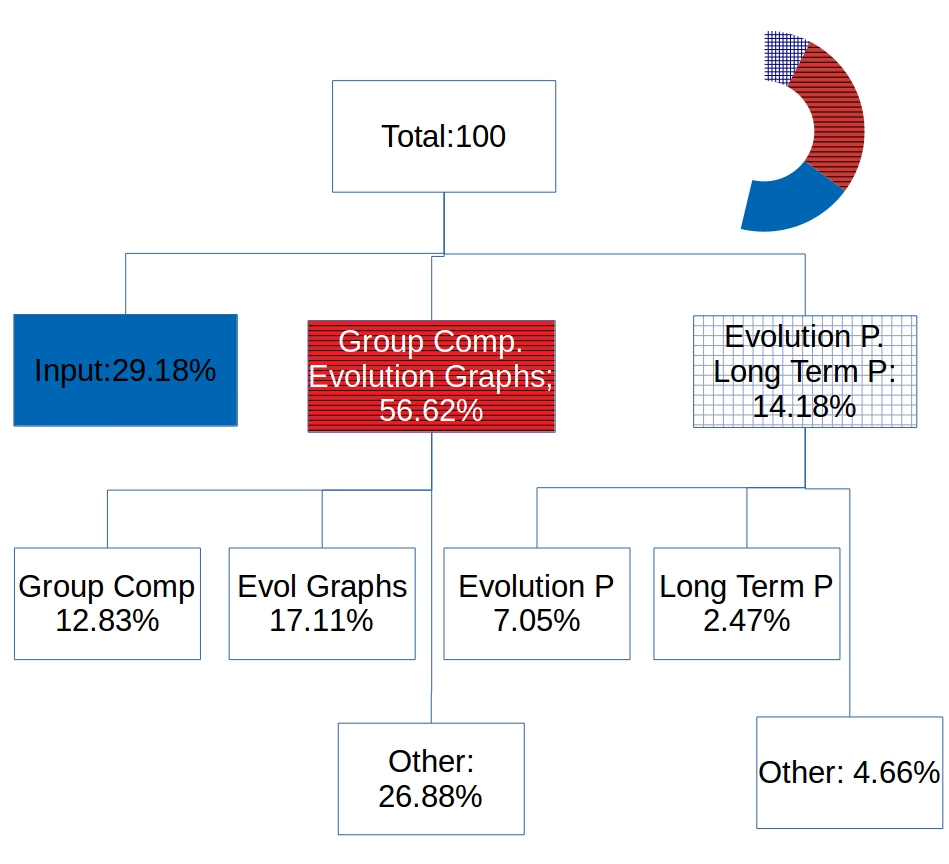} \\
    \end{center}
    \caption{Runtime Analysis, Synthetic data. Average percentage of runtime over all execution are presented. Left, times including sorting data, Right, detailed costs of the running times excluding data sorting. }
    \label{tab:RuntimeSint}
\end{figure}

Thus, we divide the execution of the algorithm in 4 large parts. First the generation and sorting of the data, which took roughly 46\% of the computation time on average. The second part corresponds to the reading of the generated data, parsing through a text file and storing the information in the variables that our algorithms uses. In this case, these variables encapsulated the events that happened every time a runner run through one control point. This took roughly 19\% of the execution time. Concerning the algorithms described throughout the paper, the part described in Algorithm~1 needed 27\% of the total execution time. This included the processing of the events in order to determine the groups at each control points and also establish the evolution graphs that allowed for group relation computation. Finally both evolution and long term patterns were computed from the evolution graphs taking a further 8\% of the time. This part of the algorithm was detailed in Algorithm~2. Note that the sorting stage is not considered in our
theoretical complexity analysis because it is only needed to the fact that we do not obtain data in real time. The theoretical complexity of the sorting stage is $O(n\delta \log n\delta)$ which is larger than the complexity of any of the remaining parts of the algorithm. This corroborates the percentages presented in figure \ref{tab:RuntimeSint}.

Actually, our theoretical analysis corresponds to the right part of figure \ref{tab:RuntimeSint} presents the average percentage of execution time once the sorting step has been excluded. Of the roughly 56\% of the time dedicated to input, event processing and pattern determination, Input needed 29.18\%, the processing of events in order to determine groups at each control point and compute the evolution graphs took 56.2\% and the determination of evolution and long term patterns took 14.18\% of the time on average. The large amount of the time needed for data input shows how efficient our algorithms are. Determining groups and computing the graphs that essentially encapsulate all the information needed to compute the evolution and long-term patterns roughly needs double the time needed to read the input data from plain text files. Within this 56.2\% of execution time excluding sorting, the determination of groups in each control point needs slightly less time (12.83\%) than the computation of evolution graphs (17.11\%). It is important to notice that the largest portion of time needed in this step was used for purposes other that the actual algorithmic computations. These included the management of the data structures used (reserving and freeing memory, pointer set-up) as well as calls to function and storage of results.

Detecting evolution and long term patterns took 14.18\% of the time. The fact that this is less than the time needed to read the input is another indicator of algorithm efficiency. This was also to be expected as these computations borrow heavily from the information computed in the previous step. Precisely, the computation of evolution patterns took about half the time needed for this step (7.05\% of the time excluding sorting) and the computation of the long-term patterns ($lpS$, $lpF$, $lpB$ and $lpR$) needed about one third of that at 2.47\%.

\subsection{Experiments with real data}\label{realEXP}
In this section, we present experiments with real data. In this case the running times of our algorithms are really small, never exceeding 2 seconds. We pay attention to the role and impact of each parameter ($m$, $\mu$ and $\varepsilon$) in the obtained results and also analyze whether their influence in the detected patterns is consistent with our theoretical analysis. As before we also analyze the time spent in each part of the algorithm providing the runtime analysis.

    We have used the freely available data~\url{https://github.com/llimllib/bostonmarathon} from the Boston Marathon (\url{http://www.baa.org/}), corresponding to 2013 and 2014 with information from 16 056 and 31 984 runners, respectively. In each case, for every runner we have information on the time (in hours, minutes and seconds) they cross each of the 12 control points set up throughout the marathon. While it must be noted that this data is relatively sparse, it is also the best that can be sourced nowadays. During our research we did find some small sets of slightly denser data (about 180 control points during a marathon) but at the price of having very few runners. To focus on the group evolution patterns introduced, we decided to use the data corresponding to the race with the largest number of runners.  We ran more than 700 experiments with different values of the $\varepsilon$, $m$, $\mu$ parameters. In all cases, computations took less than two seconds (with an average ratio of approximately 180,000 events/s). These experiments were carried out with two main goals in mind: 1) to demonstrate how our algorithm behaved with real data, and 2) to show how the three parameters mentioned conditioned the output, thus allowing for the user to customize the behavior of the algorithm.

    Concerning the influence of the parameters, special focus was put on the $\varepsilon$ parameter. Figure \ref{tab:RealData} contains information about this experiment. The columns in the figure correspond to the two editions of the Boston marathon (2013 and 2014). The top row presents data on the number of occurrences for most group behaviors as well as LS values for $\varepsilon$ values ranging from 0 to 10. The bottom row attempts to provide an overview of the effect parameter $\varepsilon$ has on the number of groups by presenting the number of occurrences of {\it appears/disappears} from $\varepsilon$ ranging from 0 to 100. Notice that, since the precision of our measurement is in the level of seconds, the value $\varepsilon=0$ will group athletes that pass through a control point during the same clock second.

\begin{figure}[h]
    \begin{center}
    \includegraphics[width=0.4\textwidth]{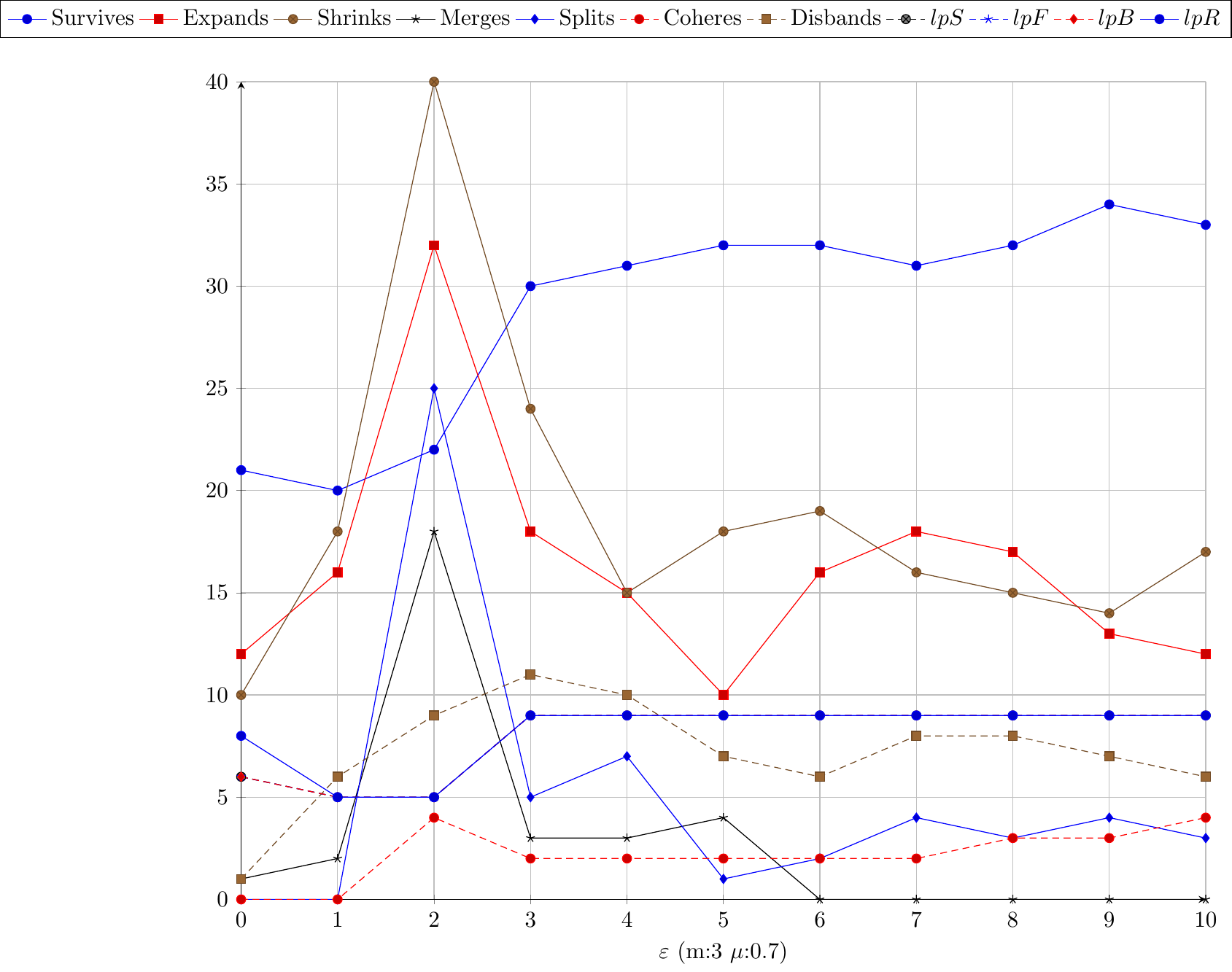} \qquad
    \includegraphics[width=0.4\textwidth]{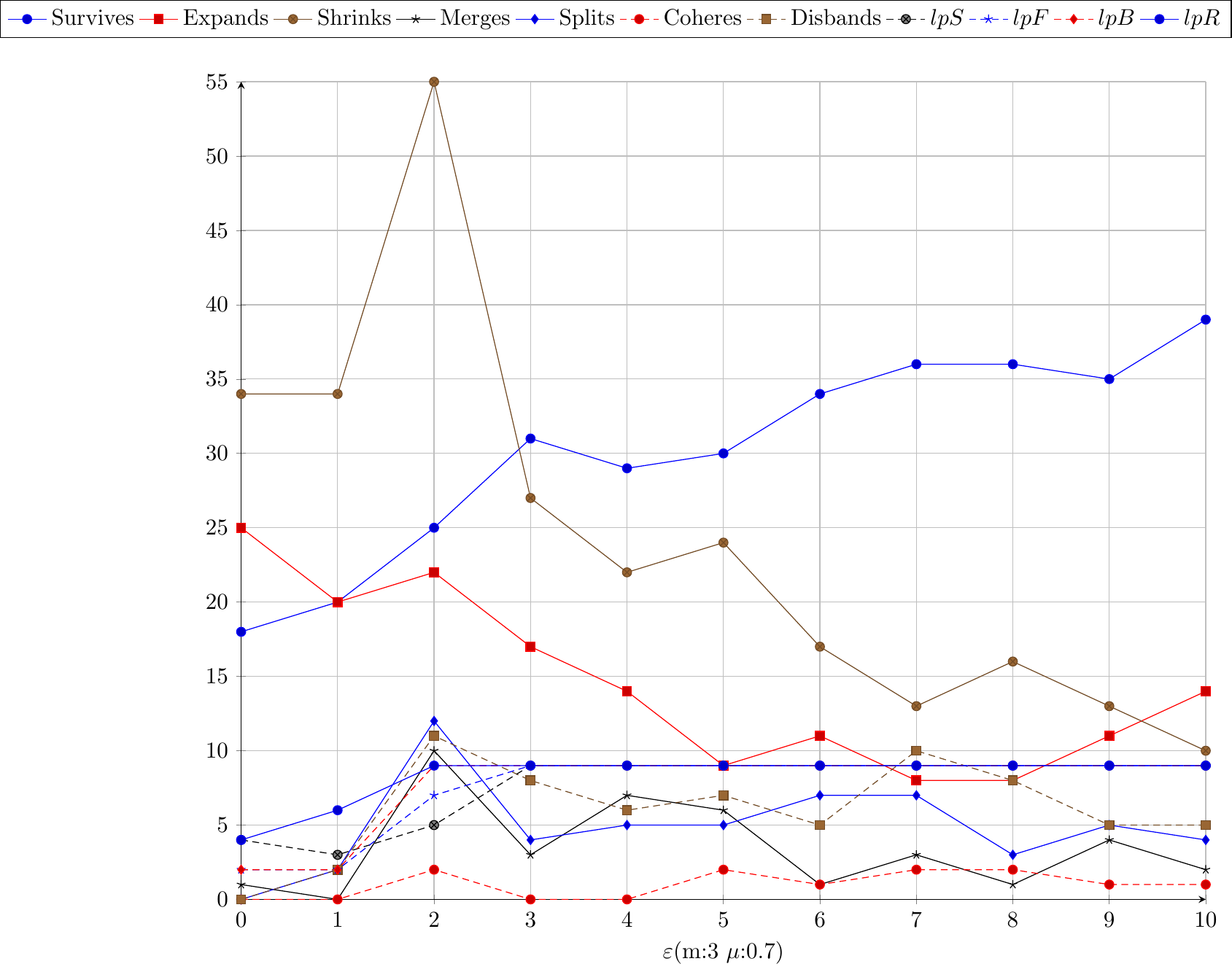} \\
\; \qquad    \includegraphics[width=0.34\textwidth]{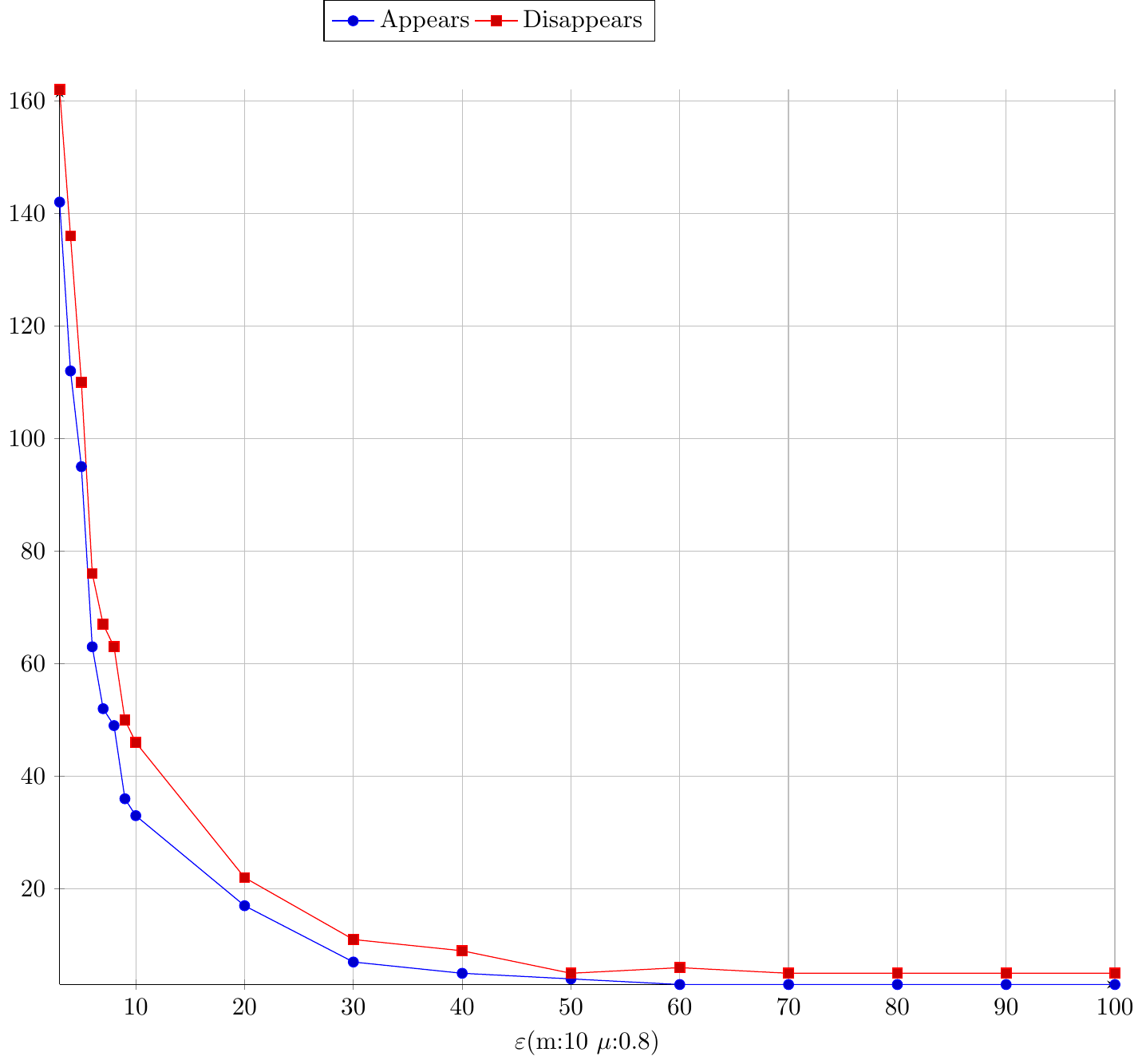} \qquad \qquad \; \;
    \includegraphics[width=0.34\textwidth]{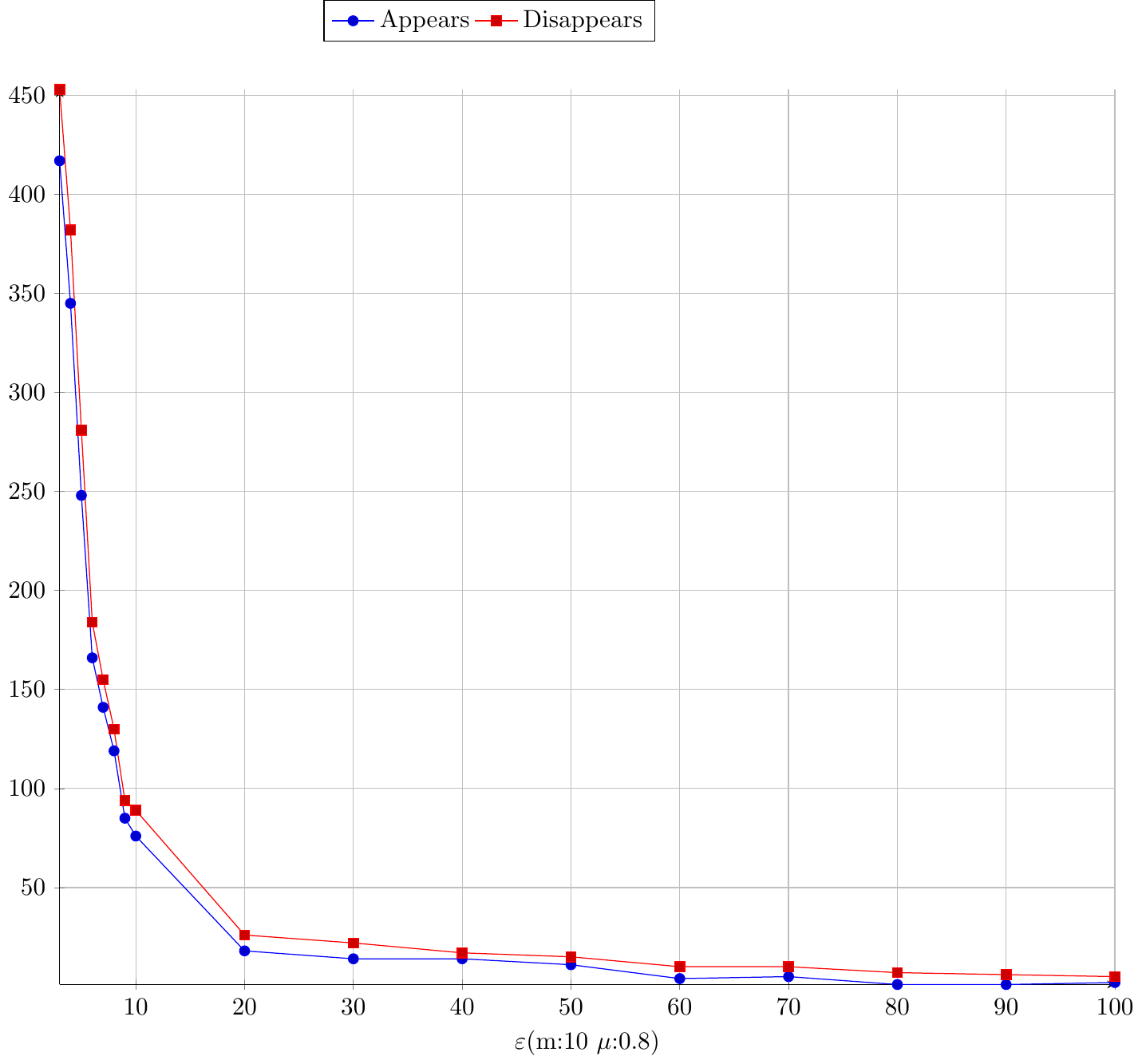}
    \end{center}
    \vspace{-1.5em} \caption{Result summary from the real data experiment. Figures in the left column contain data from the 2013 Boston Marathon and those in the right columns from 2014. Figures in the top row contain information on the number of occurrences of most group behaviors (also including $lpS$, $lpF$, $lpB$, and $lpR$) for $\varepsilon$ values from 0 to 10. Figures in the bottom row present information regarding the {\it Appears} and {\it Disappears} patterns for a longer range of values of $\varepsilon$ (from 0 to 100). }
    \label{tab:RealData}
    \vspace{-1em}
\end{figure}

The graphs in Figure \ref{tab:RealData} (bottom row), show how larger values of epsilon result in a smaller number of groups being detected. This is expected, since we allowed a longer delay between athletes to belong to the same group. In the most extreme values of $\varepsilon$ (when we allow a delay of 30 seconds between athletes in the same group) most of the race consists of a single group. Such a large value of $\varepsilon$ does not capture the natural intuition of a group of runners. If we consider more interesting smaller values of $\varepsilon$, we observe that the number of groups disappearing is larger than those appearing. This implies that the {\it Splits} evolution pattern dominated over the {\it Merges}.

The fact that we observe more splits than merges is coherent with intuition. For race organization purposes, runners are divided into several waves of runners depending on their previous running paces. Runners, thus, start grouped in several very large groups. Each of these waves will then split into smaller groups during the race. This pattern should hold for much of the race. Towards the end, the groups should stabilize as runners find their pace and make the way to the end of the race in much more resilient groups. %

As expected, as the value of $\epsilon$ decreases, more and more groups are observed. The maximum number of {\it Appears/Disappears} evolution patterns reported happen for $\epsilon=0$ (recall that due to errors in precision, two athletes will belong to the same group only when their difference in time is less than a second). With the other evolutions and long-term patterns studied in this paper (top row of Figure \ref{tab:RealData}), larger values of $\varepsilon$ produce larger numbers of {\it Survives} patterns and, subsequently, smaller occurrences of the other types of values (this is also expected, since low values of $\epsilon$ means that groups contain fewer runners, and thus it is harder find relationships between the different groups). The number of different evolution patterns observed peaks at $\varepsilon=2$, striking a balance between the different evolution patterns studied. From then on, the {\it Survives} dominates the rest and the maximum lengths of long-term patterns reach their maximums.

The figure also shows how the characterization of the evolution patterns described is greatly affected by the choice of $\varepsilon$. In televised cycle races it is common to use extremely low tolerance (say, a few milliseconds). However, if we are tracking runners in a marathon race, gaps of 3 seconds are not considered to be significant, whereas larger advantages (say, 5 seconds) could consider the runner to be out of the group. It remains for the user to determine exactly which values better capture the intuition of when two athletes are to be considered to be connected in each particular application.

To sum up, the evolution patterns presented in this paper not only concur with this intuition, but also provide specific data on that helps quantify and interpret it. For instance, the majority of {\it Splits} (a little over 72\% for a typical example shown in the figure with $\varepsilon=1$) happen in the first half of the race. This is also visible by examining the number of runners in the largest group at each control point. The number is expected to decrease, but our algorithm can determine by how many. In the previous example, the largest group starts with 611 runners and then decreases to 324, 229, 192, 137, 99, 89, 109, and 78 at each of the control points. Hence, apart form detecting patterns, the obtained results allow to make in deep analysis of the races and their evolution.

\paragraph{Runtime analysis}

Regarding the runtime performance of our algorithms for the case of real data, we have followed the same approach as for synthetic data. In this case we did not need to spend time generating data but we still needed to sort the events according to the time they happened at. This resulted in the relative weight of this part dropping from the 46\% of the synthetic case to roughly 35\%. Consequently, the other parts of the algorithm rose to about 21\% for the input of data, approximately 23\% for group determination and evolution graph computation and about 20\% for calculations related to Evolution and Long term pattern determination, figure \ref{tab:RuntimeReal}, left. Once rescaled to not include the sorting of events, the input part needed 32.68\% of the average time, the determination of groups and computational of evolution graphs took 35.67\% and the determination of evolution and long-term patterns needed 31.66\% of the average time, figure \ref{tab:RuntimeReal}, right.

\begin{figure}[h]
    \begin{center}
    \includegraphics[width=0.35\textwidth]{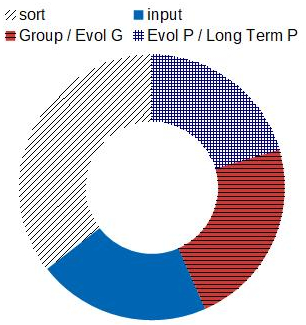} \qquad
    \includegraphics[width=0.45\textwidth]{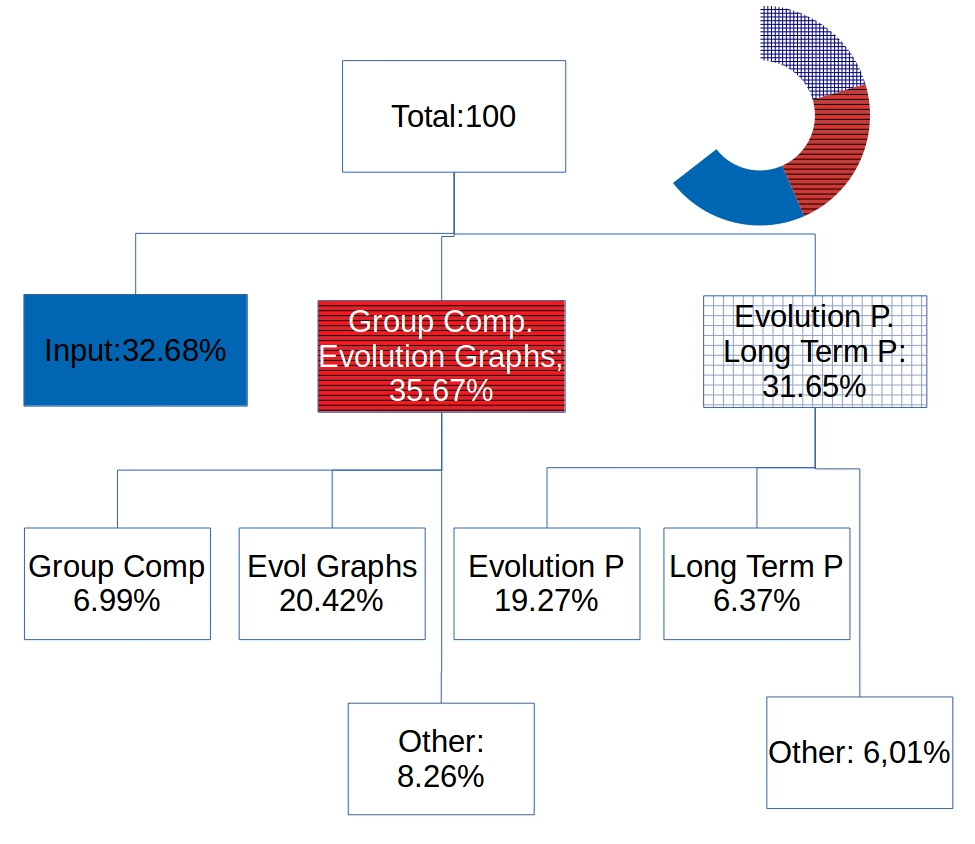} \\
    \end{center}
    \caption{Runtime Analysis, Real data. Average percentage of runtime over all execution are presented. Left, times including sorting data, Right, detailed costs of the running times excluding data sorting.}
    \label{tab:RuntimeReal}
\end{figure}

The indicators of the efficiency of our algorithms observed for synthetic data carry to the case of real data too. In this case this fact is even more evident as the time to read the data amounts to basically the same time needed to either determine groups and compute evolution graphs or that of determining patterns (both evolution and long term). A major difference to be observed respect to the synthetic case is the apparent increase in importance of the calculation of patterns respect to the computation of evolution graphs. This is likely to be caused by the data sizes. Precisely, The group determination is now 6.99\%, which is added to the 20.42\% of evolution graph computations and the 8.26\% of "other" computations for a total of 35.67\%. This is a noticeable decrease from
the 56.62\% of the synthetic case, and is mostly due to the almost 20\% decrease in the time needed by the "other" part that accounts for memory and variable management. This is very likely happening because the size of data involved in the real examples is much smaller than in the synthetic ones. Specifically, the approximately 16000 runners in the Boston Marathon are less than the number of runners in all but one of the synthetic experiments and are very far from the maximum values of two and a half million "virtual" runners. A similar effect is likely present due to the reduced number of control points. This indicates that there is probably room for improvement in the code in order to search and update the control points both as containers of runners and as parts of the race history of every particular runner. In the current code both entities are implemented using linear structures (lists) and could benefit from more efficient structures such as Hash tables. However, as the current implementation can already process events much faster than they would be produced, whether or not it needs to be optimized remains as an issue to be decided during an eventual process to port the code to real-time use. Finally, the computation of patterns gains importance as a consequence of these issues but maintains a similar distribution among its parts. The 31.65\% time it needs is distributed in 19.27\% which is again about triple of the 6.37\% of time needed for the computation of long-term patterns.

\section{Conclusions} \label{conclusions}

In this paper we presented what, to the best of our knowledge, is the first analysis of group evolution patterns that works specifically in running races or similar one-dimensional settings in which a set of entities move along a predefined course. The definitions and algorithms presented here can provide meaningful insights not only for runners, but also for race organizers and spectators or other research areas. For instances, it could be used, in the future, in crowd simulation \cite{Thal16} to simulate the movement of a large number of runners in a race with a predefined course in video games or films so that the virtual runners behave realistically. Since we provide fast algorithms, they can be run either in real-time during a race or as an analysis tool once the race is over which makes that they have even more potential applications.

First, we defined the notion of a group of athletes crossing a control point. Then, we studied the relationships between groups at two consecutive control points which, in turn, allowed us to define several evolution patterns: survives, appears, disappears, expands, shrinks, merges, splits, coheres and disbands.  We also analyzed the evolution of groups along several consecutive control points to detect long-term patterns, such as: surviving, traceable forward, traceable backward and related forward and backward, determining, every time that they appeared, the durability of each pattern. Finally, we recorded the longest appearance of each of the long-term patterns during the race.

To illustrate the efficiency of our algorithms, in Section \ref{realEXP} we have shown how currently available data (from the 2014 Boston marathon with roughly 360 000 events) can be processed in less than two seconds. Finally, we have shown how the algorithms presented in this paper are already able to process denser data (as will happen if the amount of data available continues to increase at its current rate). Specifically, Section \ref{syntheticEXP} shows how, even in cases of extremely dense data (see for example the last row of Table 1), our algorithms are still able to process more than $10^{-4}$ events/second  and thus, run in real-time. This shows their potential to be integrated into practical applications without the need for expensive equipment.

A natural continuation of this work would be to investigate the possible implementation improvements described in section \ref{realEXP}. Additionally, other situations in which our approach can be used will also be considered. Our algorithms and definitions rely heavily on the 1-D of the problem, so it would be hard to extend the same approach to arbitrarily high dimensions. However, we believe that they can be used when the data lies somewhere between one and two dimensions (e.g. the movement of people at designated metro stations, or social networks etc.).

\section{Acknowledgements}

Y. Diez is supported by the ImPACT Tough Robotics Challenge project through the Council for Science and Technology Agency, Japan. M. Fort and J. A. Sellar\` {e}s are partially funded by the MPCUdG2016-031 from the UdG. M. Korman was supported by  MEXT KAKENHI No.~17K12635 and the NSF award CCF-1422311.

\bibliographystyle{amsplain}

\end{document}